\documentclass[acmsmall,screen,nonacm]{acmart}

\acmJournal{PACMPL}
\acmVolume{1}
\acmNumber{OOPSLA} 
\acmArticle{1}
\acmYear{2022}
\acmMonth{1}
\acmDOI{} 
\startPage{1}

\setcopyright{acmcopyright}
\copyrightyear{2022}
\acmYear{2022}
\acmDOI{10.1145/1122445.1122456}

\usepackage{listings}
\usepackage{amsthm}
\usepackage{enumitem}
\theoremstyle{definition}
\newtheorem{definition}{Definition}[section]
\newtheorem{exmp}{Example}[section]

\usepackage{amsmath}
\usepackage{proof}
\usepackage{mathpartir}
\usepackage[normalem]{ulem}

\usepackage[plain]{algorithm}
\usepackage{algorithmicx}
\usepackage[noend]{algpseudocode}

\usepackage{xspace}

\usepackage{caption}
\usepackage{subcaption}
\usepackage[toc,page]{appendix}

\newcommand{\tool}{{\sc NGST2}\xspace}

\newcommand{\toolasnbidirec}{{\sc ASN-Bidirec}\xspace}

\newcommand{\toolnoprune}{{\sc NGST2-NoPrune}\xspace}
\newcommand{\toolnotrace}{{\sc NGST2-NoTrace}\xspace}

\newcommand{\net}{{\sc CGN}\xspace}
\newcommand{\asn}{{\sc ASN}\xspace}
\newcommand{\seqtoseq}{{\sc SEQ2SEQ}\xspace}
\newcommand{\cegisneo}{{\sc CEGIS-Neo}\xspace}
\newcommand{\enumasn}{{\sc Enum-ASN}\xspace}
\newcommand{\enumseqtoseq}{{\sc Enum-Seq2Seq}\xspace}
\newcommand{\netbeans}{{\sc Netbeans}\xspace}

\newcommand{\idsl}{IMP\xspace}
\newcommand{\fdsl}{$\lambda_{str}$\xspace}

\newcommand{\grammar}{\mathcal{G}}

\newcommand{\code}[1]{\texttt{#1}\xspace}

\newcommand{\todo}[1]{\textcolor{red}{\textbf{TODO:} #1}}

\newcommand{\pexp}{\tilde{e}}
\newcommand{\pcand}{\tilde{C}}

\newcommand{\progeq}[2]{{#1} \equiv {#2}}
\newcommand{\udjudge}[3]{{#1} \vdash {#2} \updownarrow {#3}}
\newcommand{\upjudge}[3]{{#1} \vdash {#2} \uparrow {#3}}

\newcommand{\exec}[2]{\llbracket {#1} \rrbracket_{#2}}

\newcommand{\inv}[2]{\llbracket {#1} \rrbracket^{#2}}
\newcommand{\inva}[3]{\llbracket {#1} \rrbracket^{#2}({#3})}
\newcommand{\consist}{\backsimeq}
\newcommand{\consista}[2]{{#1} \consist {#2}}
\newcommand{\nconsist}{\nbacksimeq}
\newcommand{\nconsista}[2]{{#1} \nconsist {#2}}
\newcommand{\conjudge}[3]{{#1} \vdash \consista{#2}{#3}}

\newcommand{\nconjudge}[3]{{#1} \vdash \nconsista{#2}{#3}}
\newcommand{\derives}{\Rightarrow^*}
\DeclareMathOperator*{\argmax}{arg\,max}

\newcommand{\prog}{\mathcal{P}}
\newcommand{\nn}{\mathcal{M}_\theta}
\newcommand{\worklist}{\mathcal{W}}
\newcommand{\cexs}{\mathcal{E}}
\newcommand{\symexp}{\psi}
\newcommand{\symexec}[1]{\llbracket {#1} \rrbracket^{\sharp}}
\newcommand{\symexecback}[2]{\llbracket {#1} \rrbracket^{\sharp^{#2}}}
\newcommand{\lang}{\mathcal{L}}
\newcommand{\valtyp}{\mathcal{V}}
\newcommand{\vartyp}{Var}
\newcommand{\derive}[1]{\mathcal{R}({#1})}
\newcommand{\astpath}{\pi}
\newcommand{\linear}{l}
\newcommand{\lstm}{\mathrm{LSTM}}
\newcommand{\lstmcopy}{\mathrm{LSTM}_{\mathrm{copy}}}
\newcommand{\lstmstar}{\mathrm{LSTM}_{\star}}
\newcommand{\bilstm}{\mathrm{BiLSTM}}
\newcommand{\bilstmcopy}{\mathrm{BiLSTM}_{\mathrm{copy}}}
\newcommand{\bilstmstar}{\mathrm{BiLSTM}_{\star}}
\newcommand{\ffnn}{\mathrm{FFNN}}
\newcommand{\softmax}{\mathrm{softmax}}
\newcommand{\attn}{\mathrm{Attn}}

\newcommand{\encode}{\mathrm{enc}}
\newcommand{\hidden}{s}
\newcommand{\coarse}{\mathcal{C}}
\newcommand{\functools}{{\tt functools}\xspace}
\newcommand{\translate}{\mathrm{translate}}
\newcommand{\defeq}{\triangleq}
\newcommand{\symexpp}{\symexp^\prime}
\newcommand{\symexpfr}{\hat{\symexp}}
\newcommand{\symexpbar}{\bar{\symexp}}
\newcommand{\binop}{\oplus}
\newcommand{\len}{\mathsf{len}}

\newif\ifExt\Exttrue

\newif\ifDiff\Difffalse

\ifDiff
\newcommand{\new}[1]{\textcolor{teal}{#1}}
\newcommand{\old}[1]{\textcolor{red}{\sout{#1}}}
\else
\newcommand{\new}[1]{{#1}}
\newcommand{\old}[1]{}
\fi

\newcommand{\coll}[3]{\mathcal{C}\llbracket {#1} \rrbracket_{#2,#3}}
\newcommand{\powerset}{\raisebox{.15\baselineskip}{\Large\ensuremath{\wp}}}
\newcommand{\sterm}{{\tau}}

\newcommand{\collctx}[3]{\mathcal{C}_{ctx}\llbracket {#1} \rrbracket_{#2,#3}}

\begin{document}

\title[Automated Transpilation of Imperative to Functional Code using Neural-Guided Program Synthesis (Extended)]{Automated Transpilation of Imperative to Functional Code using Neural-Guided Program Synthesis (Extended Version)}

\author{Benjamin Mariano}
\email{bmariano@cs.utexas.edu}
\affiliation{%
  \institution{University of Texas at Austin}
  \country{USA}    
}
\author{Yanju Chen}
\email{yanju@cs.ucsb.edu}
\affiliation{%
  \institution{University of California, Santa Barbara}
  \country{USA}  
}
\author{Yu Feng}
\email{yufeng@cs.ucsb.edu}
\affiliation{%
  \institution{University of California, Santa Barbara}
  \country{USA}  
}
\author{Greg Durrett}
\email{gdurrett@cs.utexas.edu}
\affiliation{%
  \institution{University of Texas at Austin}
  \country{USA}      
}
\author{I\c{s}il Dillig}
\email{isil@cs.utexas.edu}
\affiliation{%
  \institution{University of Texas at Austin}
  \country{USA}      
}



\begin{abstract}
While many mainstream  languages such as Java, Python, and C\# increasingly incorporate functional APIs to simplify programming and improve parallelization/performance, there are no effective techniques that can be used to \emph{automatically} translate existing imperative code to functional variants using these APIs. Motivated by this problem, this paper presents a transpilation approach based on inductive program synthesis for modernizing existing code. Our method is based on the observation that the overwhelming majority of source/target programs in this setting satisfy an assumption that we call \emph{trace-compatibility}: not only do the programs share syntactically identical low-level expressions, but these expressions also take the same values in corresponding execution traces. Our method leverages this observation to design a new neural-guided synthesis algorithm that (1) uses a novel neural architecture called \emph{cognate grammar network (CGN)} and (2) leverages a form of concolic execution to prune partial programs based on \emph{intermediate values} that arise during a computation. We have implemented our approach in a tool called \tool and use it to translate imperative Java and Python code to functional variants that use the {\tt Stream} and {\tt functools} APIs respectively. Our experiments show that \tool significantly outperforms several baselines and that our proposed neural architecture and pruning techniques are vital for achieving good results. 
\end{abstract}




\maketitle

\section{Introduction}
\label{sec:intro}

In recent years, a number of mainstream programming languages have introduced functional APIs to offer users the benefits of functional programming within imperative languages. For instance, Java 8 introduced the notion of streams, which provides an API for processing sequences of elements using common functional operators like map and filter. Similarly, Python offers an expansive functional API incorporating functional operators like map, reduce, and list comprehensions.

While these functional APIs provide a number of benefits such as convenient parallelization \cite{khatchadourian2020safe}, succinct definitions, and easier readability, converting imperative code to functional APIs can be challenging, particularly for developers accustomed to imperative languages. Recognizing this drawback, prior efforts introduced rule-based translators to automatically refactor imperative code into (equivalent) functional versions \cite{gyori2013crossing}, and such approaches have even been incorporated into mainstream IDEs like NetBeans. However, since manually crafting rule-based translators for \emph{arbitrary} imperative code is infeasible, existing approaches can only handle stylized code snippets over commonly occurring imperative code patterns.


In this paper, we propose an alternative solution to this transpilation problem based on \emph{counterexample guided inductive synthesis (CEGIS)}~\cite{solar2007sketching}. Our approach  incorporates a novel inductive synthesis engine that combines neural networks (for guiding search) with powerful bidirectional pruning rules based on concolic execution. At a high level, our solution is based on two key observations:

\begin{itemize}[leftmargin=*]
    \item While the high-level structure of the code is quite different between the imperative and functional versions, the source and target programs actually share many syntactically identical low-level expressions (e.g., arguments of API calls). We formalize this observation using a new notion that we call \emph{cognate  grammars} and leverage it to reduce the size of the search space.
    \item Beyond being syntactically identical, the source and target programs enjoy a property that we call \emph{trace-compatibility}:  shared expressions in the two programs always take the same values when executing the source and target programs on the same inputs. 
\end{itemize}

Guided by these observations, we have developed a new neural-guided inductive synthesizer called \tool\footnote{Stands for \emph{Neural Guided Source To Source Translation}} (see Figure~\ref{fig:overview}) that  (1) incorporates a new neural architecture called a \emph{cognate grammar network (CGN)} to guide the search, and (2) a powerful pruning system that significantly reduces the search space using the trace compatibility observation. Our proposed CGN model is an adaptation of the existing \emph{abstract syntax network (ASN)}~\cite{asn} (which is a top-down, tree-structured model that is suitable for reasoning about programs), but it leverages the syntactically shared terms between the source  and target programs to make more accurate predictions. We train the CGN model off-line on   pairs of equivalent (but synthetic) imperative and functional programs and then use it on-line at synthesis time to guide a top-down enumerative search engine. Given a partial abstract syntax tree (AST) where some of the nodes are non-terminal symbols in the target grammar, our approach uses the CGN to predict which grammar productions to use for expanding the non-terminals. This results in a refined AST which is then checked under the trace compatibility assumption. In particular, the goal of  {this check} is to determine whether a given partial program produces any intermediate values that are  inconsistent with the source program. If so, we can prune all completions of the partial program from the search space. In contrast to prior pruning approaches used in inductive program synthesis~\cite{feser2015synthesizing,albarghouthi2013recursive,alur2015synthesis}, the key novelty of our approach is to leverage \emph{intermediate} values (as opposed to just input-output examples) to perform more aggressive pruning.




\begin{figure}
    \centering
    \includegraphics[width=.8\linewidth]{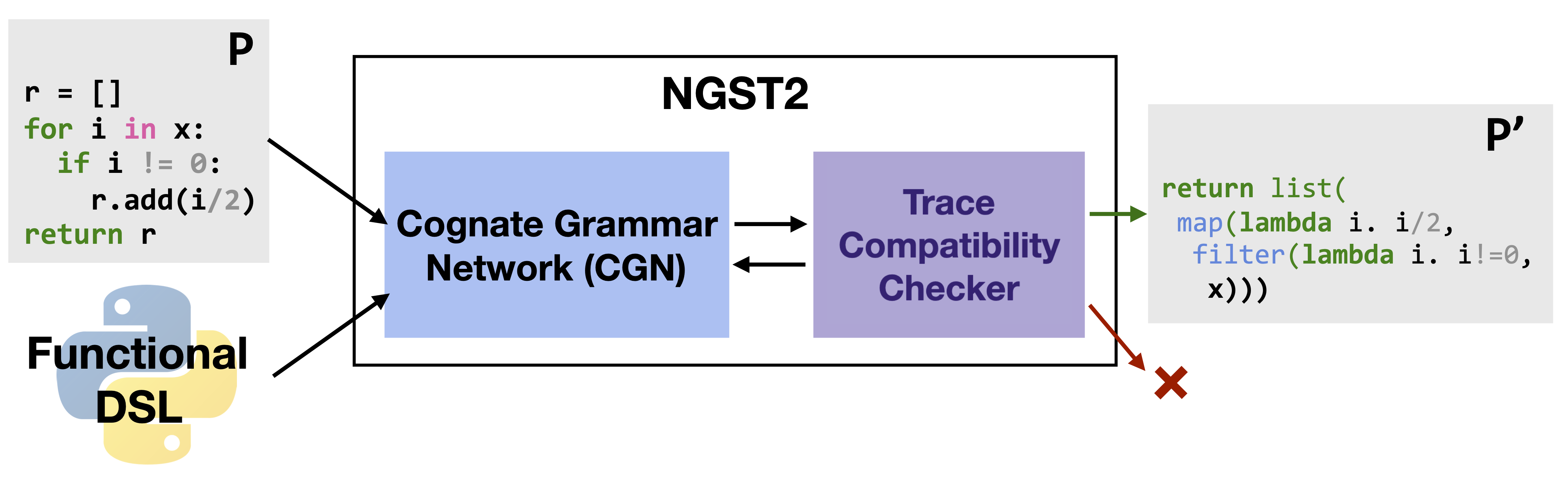}
    \vspace{-0.3in}
    \caption{Schematic Overview of \tool}
    \label{fig:overview}
    \vspace{-0.1in}
\end{figure}



We evaluate \tool in two different settings that require translating imperative code snippets to equivalent functional variants. In the first setting, we use \tool to migrate imperative Java programs to the functional Java Stream API. In our second client, we use \tool to translate imperative Python code to a functional style using the \functools API and other Pythonic constructs, such as list comprehensions. Overall, \tool is able to automate 80\% of these challenging source-to-source translation tasks and significantly outperforms existing techniques and simpler baselines/ablations. These experimental results demonstrate the effectiveness of our {trace compatibility} checking rules as well as the proposed neural CGN architecture and provide evidence that this work advances the state-of-the-art in imperative-to-functional code translation.

In summary, this paper makes the following key contributions:

\begin{enumerate}[leftmargin=*]
\item We propose the first\old{ general} \new{broadly applicable} solution for translating imperative code to equivalent versions using functional APIs.  
    \item We identify and formalize two key characteristics of this problem (namely, \emph{cognate grammars} and \emph{trace compatibility}) and propose a new neural-guided inductive synthesizer that leverages these properties.
    \item We propose a new neural architecture called \emph{cognate grammar network (CGN)} that can be used for source-to-source translation tasks involving cognate grammars with shared terms.
    \item We propose a novel  pruning technique --- based on the trace compatibility assumption --- for detecting partial programs that produce intermediate values that are inconsistent with the source program. Our pruning technique is based on concolic execution and leverages bidirectional reasoning to reduce the overhead of trace compatibility checking.  
    \item We implement our algorithm in a tool called \tool and evaluate it on real-world programs from two different domains and show that it significantly outperforms existing approaches.
\end{enumerate}
\section{Overview}
\label{sec:overview}

\begin{figure}[t]
    \centering
    \begin{lstlisting}[basicstyle=\small,linewidth=20cm,frame=none, language=Java,breaklines]
    public List<String> getUserRoles(String uID, List<Policy> policies) {
        List<String> roles = new ArrayList<>();
        for (Policy policy : policies) {
            for (Role role : (*\colorbox{pink}{policy.getRoles()}*)) {
                if ((*\colorbox{pink}{role.getIDs().contains(uID)}*)) 
                    roles.add((*\colorbox{pink}{role.getName()}*));
            }
        }
        return roles;
    }
    \end{lstlisting}    
    \vspace{-0.3in}
    \caption{Imperative Java Program}
    \label{fig:motivex_imp}
    \vspace{-0.1in}
\end{figure}

We give a high-level overview of our method with the aid of the motivating example shown in Figure~\ref{fig:motivex_imp}. Here, the \code{getUserRoles} procedure  takes as input a user ID (\code{uID}) and a list of policies and returns all the ``roles'' (e.g., beneficiary, policy owner, etc.) that the specified user has across all policies.
 Figure~\ref{fig:motivex_imp} performs this computation in an imperative style using a doubly nested loop, but it is possible to express the same program logic in a functional style using the Java Stream API. As shown in  Figure~\ref{fig:motivex_func}, the functional implementation first converts the input policy to a stream and then obtains the desired result using the higher-order combinators \code{map}, \code{filter}, and \code{flatMap}\footnote{\code{flatMap} maps each element of the input to a stream and then flattens the result to a single stream.} provided by the Stream API. 


\begin{figure}[t]
    \centering
    \begin{lstlisting}[basicstyle=\small,linewidth=20cm,frame=none, language=Java,breaklines]
    public List<String> getUserRoles(String uID, List<Policy> policies) {
       return policies.stream()
                .flatMap(policy -> (*\colorbox{pink}{policy.getRoles()}*).stream())
                .filter(role -> (*\colorbox{pink}{role.getIDs().contains(uID)}*))
                .map(role - > (*\colorbox{pink}{role.getName()}*))
                .collect(Collectors.toList());
    }
    \end{lstlisting}
    \vspace{-0.1in}
    \caption{Java Stream Program}
    \label{fig:motivex_func}
    \vspace{-0.1in}    
\end{figure}

While the high-level structure of the code is clearly quite different between Figures~\ref{fig:motivex_imp} and~\ref{fig:motivex_func}, we notice 
 that many of the low-level expressions (highlighted in pink) are actually shared between the two programs. In fact, on deeper 
inspection, we realize that the relationship between these expressions goes beyond just surface syntax. As shown in the sample execution traces in Figures~\ref{fig:impexec} and ~\ref{fig:funcexec}, these shared expressions also take the same set of values in two corresponding executions of the imperative and functional programs.

\begin{figure}
    \centering
    \includegraphics[width=.8\textwidth]{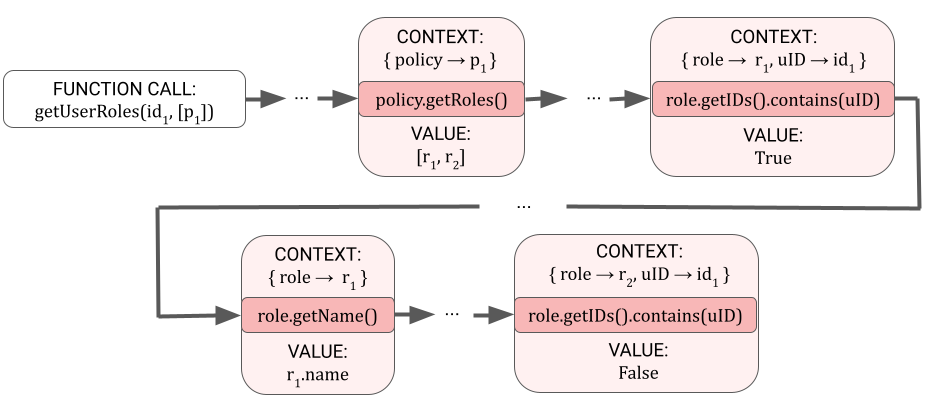}
    \caption{Execution of Imperative Program through Shared Expressions}
    \label{fig:impexec}
    \vspace{-0.1in}
\end{figure}

As illustrated by this example, code snippets written using functional APIs  exhibit close {syntactic} \emph{and} semantic ties to their corresponding imperative version. Our method takes advantage of such relationships between the two programs to dramatically simplify the underlying synthesis task. In particular, we formalize the syntactic relationship between the source and target programs through the concept of \emph{cognate grammars} and express their semantic ties using the notion of \emph{trace compatibility}. 

\subsection{Leveraging Syntactic Ties via CGN}
Intuitively, the syntactic relationship between the two program versions is extremely useful for restricting the search space that a program synthesizer needs to explore. In particular, once the synthesizer produces  a ``sketch'' of the functional program with a certain choice of functional combinators, the arguments of those combinators are not \emph{arbitrary} expressions but rather small snippets taken from the original imperative program.  Thus, this observation immediately gives us a way to reduce the search space.

\begin{figure}
    \centering
    \includegraphics[width=.8\textwidth]{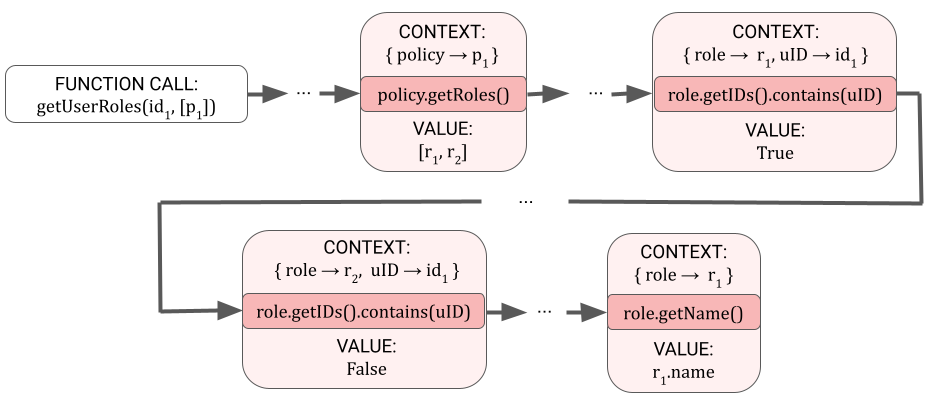}
    \caption{Execution of Functional Program through Shared Expressions}
    \label{fig:funcexec}
    \vspace{-0.2in}
\end{figure}

In addition, our method also uses the syntactic relationship between the two programs to more effectively \emph{guide}  search.
In particular, similar to many prior techniques~\cite{balog2016deepcoder,ye2020optimal}, we want to leverage a neural model  to predict more likely target programs based on the source program. However, existing techniques for neural-guided synthesis do not exploit the syntactic similarity between source and target expressions and predict \emph{grammar productions} instead of \emph{concrete expressions} from the source program. In contrast, we need a method that can predict grammar productions when generating the high-level structure of the program but that \emph{switches to predicting source expression when filling in the remaining holes}.  

To address this problem, we formulate the source and target languages as a pair of \emph{cognate grammars} where some of the non-terminals are \emph{shared} between the two languages. Then, based on this concept, we propose a new neural architecture called   \emph{cognate grammar network (CGN)}  to perform predictions over such  grammars. Our proposed architecture extends the existing \emph{abstract syntax network} model \cite{asn} (for generating syntactically valid programs) with  a so-called \emph{pointer network} \cite{jia-liang-2016-data,see-etal-2017-get} in order to make useful predictions when expanding shared non-terminals.


\subsection{Leveraging Semantic Ties}

\begin{figure}
\begin{center}
\centering
  \begin{minipage}{\linewidth}
\begin{lstlisting}[linewidth=6.5cm,frame=none, numbers=none, language=Java]
return (*$?_0$*).flatMap(policy -> (*\colorbox{pink}{$?_1$}*))
         .map(role -> (*\colorbox{pink}{$?_2$}*))
\end{lstlisting}
\end{minipage}
\end{center}
\vspace{-0.2in}
\caption{Partial program encountered during search.} \label{fig:partial}
\vspace{-0.2in}
\end{figure}

While our proposed CGN architecture makes it possible to exploit syntactic similarities between the source and target programs, we also want to exploit the \emph{semantic ties} between them.  To see what we mean by this, consider the partial program shown in Figure~\ref{fig:partial}. Here, $?_0$, $?_1$, and $?_2$ represent unknown expressions, and the fact that $?_1$ and $?_2$ are highlighted in pink indicates that they must be instantiated with one of the highlighted expressions from the source program in Figure~\ref{fig:motivex_imp}.  Even though $?_0$ is completely unconstrained (i.e., it can be any arbitrary expression), we can use our trace compatibility assumption to prune this partial program using the following chain of inferences for the program executions pictured in Figures~\ref{fig:impexec} and \ref{fig:funcexec}:

\begin{itemize}[leftmargin=*]
    \item Among the four source expressions that $?_1$ could be instantiated with, the only one that would even type check is \code{policy.getRoles()}, which only takes on  value $[r_1, r_2]$ in the execution trace shown in Figure~\ref{fig:impexec}. 
    \item Thus, given the semantics of \code{flatMap}, we can conclude that the expression $?_0.$\code{flatMap(...)} must produce a list which contains $[r_1, r_2]$ as a sublist when executed on the same input as the imperative program.
    \item Now, reasoning backwards from the return value, we know that the output for the partial program on this trace needs to be the list $[r_1]$. \item But, no matter how we instantiate $?_2$, we can never obtain the list $[r_1]$ when invoking \code{map} on an input list containing $[r_1, r_2]$ as a sublist (as \code{map} is length preserving). Thus, we conclude that no completion of this partial program can be equivalent to the imperative source program from Figure~\ref{fig:motivex_imp}.
\end{itemize}

Note that this kind of reasoning would \emph{not} be possible without knowing intermediate values in the execution of the source program. In particular,  if we did not know the possible values that $?_1$ could take, we would have to assume that \code{map} could be called on arbitrary list, which does not provide any pruning opportunities. Our synthesis engine leverages  such shared intermediate values between the two programs and prunes infeasible partial programs using a form of concolic execution.

\section{Preliminaries and Problem Statement}
\label{sec:prelim}

In this section, we first present some preliminary information and then formalize our problem.

\subsection{CFGs and Partial Programs}

As standard, we define the syntax of our source and target languages in terms of a context-free grammar (CFG):
\begin{definition}[{\bf CFG}]
A context-free grammar $\mathcal{G}$ is a tuple $(V, \Sigma, R, S)$
where $V$ are non-terminals, $\Sigma$ are terminals, $R$ are 
productions, and $S$ is the start symbol.
\end{definition}

Given a string $s \in (\Sigma \cup V)^*$, we use the notation $s \Rightarrow s'$ to denote that $s'$ is obtained from $s$ by replacing one of the non-terminals $N$ in $s$ by a string $w \in (\Sigma \cup V)^*$ such that $N \rightarrow w$ is a production in the grammar. We use the notation $\Rightarrow^*$ to denote the transitive closure of $\Rightarrow$. 

\begin{definition}[{\bf Partial program}]
Let $\mathcal{G}=(V, \Sigma, R, S)$ be a context-free grammar for some programming language $\lang$. A \emph{partial program}  in $\lang$ is a string $\prog \in (\Sigma \cup V)^*$ such that $S \Rightarrow^* \prog$. We  say that $\prog$ is \emph{complete} if $\prog$ does not contain any non-terminals. Finally, a complete program $\prog'$ is a \emph{completion} of partial program $\prog$ if $\prog \derives \prog'$.
\end{definition}

\subsection{Program Equivalence and Transpilation}

The problem we consider in this paper is an instance of \textit{transpilation}, where the goal is to translate a program $\prog_1$ in a programming language $\lang_1$ to a \emph{semantically equivalent} program $\prog_2$ in a different language $\lang_2$. In this work, we presume program semantics are given in terms of denotational semantics {$\exec{\cdot}{\sigma} : \Sigma^* \rightarrow \valtyp$} which maps a complete program {$\prog \in \Sigma^*$} to a value $v \in \valtyp$ given valuation $\sigma: \vartyp \rightarrow \valtyp$  mapping variables in $\prog$ to values.

\begin{definition}[{\bf Equivalence}]
Two programs $\prog_1, \prog_2$ are \emph{semantically equivalent}, denoted $\progeq{\prog_1}{\prog_2}$, if, for all input valuations $\sigma$, we have $\exec{\prog_1}{\sigma} = \exec{\prog_2}{\sigma}$.\footnote{{Since we out technique does not target reactive programs, we assume all source programs terminate, so our equivalence definition is for terminating programs.}} 
\end{definition}

In other words, we consider two programs to be equivalent if they always produce the same output when executed on the same input. 

\begin{definition}[{\bf Transpilation problem}]
\label{def:transpile}
Let $\prog_s$ be a program in some source language $\lang_s$ defined by grammar $\grammar_s$. Given a target language $\lang_t$ defined by grammar $\grammar_t$, the transpilation problem is to find a (complete) program $\prog_t$ in $\lang_t$ such that $\progeq{\prog_s}{\prog_t}$.
\end{definition}

\subsection{Trace Compatible Transpilation}

As discussed in Section~\ref{sec:overview}, the  transpilation problem as introduced in Definition~\ref{def:transpile} is too general to be amenable to an efficient solution (at least without placing severe syntactic restriction on the source or target languages). Motivated by a property we have observed in our target application domain (namely, translating imperative code to functional APIs), we first introduce the \emph{trace compatibility} assumption and then refine the transpilation problem under this restriction.

First, since the notion of trace compatibility only makes sense for certain combinations of languages, we  introduce the concept of \emph{cognate grammars}:

\begin{definition}[{\bf Cognate grammars}]
\label{def:cognate}
Grammars $\mathcal{G}_1 = (V_1, \Sigma_1, R_1, S_1)$ and   $\mathcal{G}_2 = (V_2, \Sigma_2, R_2, S_2)$  are \textit{cognate} iff, for every \emph{shared non-terminal} $N \in {V}_1 \cap {V}_2$, we have    $N \derives  \sterm$ in $\mathcal{G}_1$ iff $N \derives  \sterm$ in $\mathcal{G}_2$. We refer to the words $ \sterm \in (\Sigma_1 \cap \Sigma_2)^*$ that are derivable from the shared non-terminals as \emph{shared terms} between $\mathcal{G}_1$ and $\mathcal{G}_2$.
\end{definition}

In other words,  two grammars are said to be \emph{cognate} if  every shared non-terminal symbol $N \in V_1 \cap V_2$ can derive the same words from $N$ in both grammars. As mentioned in Section~\ref{sec:intro}, this notion of cognate grammars makes sense in the context of translating imperative code to functional APIs because, while the high-level structure of the source and target programs may differ, the low-level expressions often tend to be shared. Thus, we can structure the grammar of the source and target languages to be cognate --- i.e., low-level expressions of the same type are derived from shared non-terminal symbols, whereas high-level constructs (e.g., {\tt for} loops, higher-order combinators, etc.) are derived from unshared non-terminals.

Beyond having  cognate source and target languages, another important  observation in our  setting is that shared terms almost always take the same values in corresponding pairs of executions. To formalize this observation which we refer to as \emph{trace compatibility}, we first assume  a collecting semantics \cite{hudak1991collecting, cousot1994higher} $\coll{\cdot}{\prog}{\sigma} : Expr \rightarrow \powerset(\mathcal{V})$ which maps an expression $e$ to the set of values which $e$ takes during execution of $\prog$ on $\sigma$.

 \begin{example}
 Consider the following  program $\prog$ that takes as input an integer $n$:

\begin{center}
\centering
\begin{minipage}{0.35\linewidth}

\begin{lstlisting}[numbers=none,linewidth=5.5cm,breaklines]
for(i :=0; i < n; i+=2) {
    print(i+1);
}
\end{lstlisting}
\end{minipage}
\end{center}
Here, for input valuation $\sigma: [n \mapsto 3]$, we have
$\coll{i+1}{\prog}{\sigma} = \{ 1, 3 \}$ since these are the values that $i+1$ takes during this execution. 
 \end{example}
 
Next, we define \emph{trace compatibility} as follows:

\begin{definition}[{\bf Trace compatibility}]
 Let $\prog_s$ and $\prog_t$ be a pair of source and target programs from cognate  grammars $\grammar_s$ and $\grammar_t$. We say that $\prog_s$ and $\prog_t$ are \emph{trace-compatible} on $\sigma$, denoted $\conjudge{\sigma}{\prog_s}{\prog_t}$  iff for every shared term $\sterm$ \new{appearing (at least once)} in $\prog_t$ we have (1) $\sterm$ \new{appears (at least once)} in $\prog_s$, and (2) $\coll{\sterm}{\prog_t}{\sigma} \subseteq \coll{\sterm}{\prog_s}{\sigma}$.
\end{definition}
 
 In other words, trace compatibility of $\prog_s$ and $\prog_t$ on input $\sigma$ means that (a) every shared term $\sterm$ that appears in $\prog_t$ also syntactically appears in the source program $\prog_s$, and (b) for every value that $\sterm$ can take when executing $\prog_t$ on $\sigma$, there is a corresponding value of $\sterm$ when executing $\prog_s$ on $\sigma$. \new{Note that if there are multiple occurrences of $\sterm$ in some program $\prog$ (i.e., for either $\prog_s$ or $\prog_t$), then $\coll{\sterm}{\prog}{\sigma}$ contains the values that \emph{all} occurrences of $\sterm$ take during execution of $\prog$ on $\sigma$.}
 
 \begin{example}
Consider the following imperative program $\prog_s$, which prints every odd element of input list $x$ multiplied by $2$, and the functional program $\prog_t$ which prints \emph{every} element multiplied by $2$. We assume the term $i*2$ (highlighted in pink) is a shared term in the grammars of $\prog_s$ and $\prog_t$.
\begin{center}
    \begin{tabular}{cc}
\begin{minipage}{0.35\linewidth}

\begin{lstlisting}[numbers=none]
for(int $i$ : odd(x)) {
    print((* \colorbox{pink}{$i*2$}*)); 
}
\end{lstlisting}
\end{minipage}
             &
\begin{minipage}{0.4\linewidth}

\begin{lstlisting}[numbers=none]
map(x, $\lambda$ i. print((* \colorbox{pink}{$i*2$}*)))
\end{lstlisting}
\end{minipage}
    \end{tabular}
\end{center}
These two programs are clearly not equivalent, and they also violate trace compatibility on the input  $\sigma = \{x \mapsto [2]\}$. In particular, $i*2$ produces value $4$ in $\prog_t$ but not $\prog_s$.


\end{example}
 
Using this formalization of trace compatibility, we next define the \emph{trace compatible transpilation} problem that we address in the rest of this paper:

\begin{definition}{\bf (Trace compatible transpilation)}
Let $\lang_s$ and $\lang_t$ be a pair of source and target languages defined by two cognate context-free grammars. Given a program $\prog_s$ in $\lang_s$, the \emph{trace compatible transpilation problem} is to find a program $\prog_t$ in $\lang_t$ such that (1) $\prog_s \equiv \prog_t$, and (2) for all input valuations $\sigma$, we have $\conjudge{\sigma}{\prog_s}{\prog_t}$.
\end{definition}

\section{Neural-Guided Transpilation Algorithm}
\label{sec:alg}

\begin{figure}
\begin{algorithm}[H]
    \begin{algorithmic}[1]
        \Procedure{Transpile}{$\mathcal{P}, \nn, \grammar$}
            \State \textbf{input:} Source program $\mathcal{P}$
            \State \textbf{input:} Neural model $\nn$
            \State \textbf{input:} Context-free grammar $\grammar = (V, \Sigma, R, S) $
            \vspace{0.05in}
            \State $\worklist \gets \{S\}$ \Comment{initialize worklist  to $S$, the empty partial program}
            \State $\cexs \gets \emptyset$ \Comment{initialize counterexamples to empty set}
            \While{$\worklist \neq \emptyset$}
                \State $\prog' \gets \textsc{ChooseBest}(\worklist, \nn, \prog)$ \Comment{Dequeue top candidate from priority queue}

                \If{\textsf{IsComplete}($\prog'$)} 
                    \State $(\nu, \sigma) \gets $ \textsf{IsEquivalent}($\prog$, $\prog'$\new{, $\cexs$}) 
                    \If{$\nu$} \Comment{Candidate is equivalent}
                        \State \Return $\prog'$
                    \EndIf
                    \State $\cexs \gets \cexs \cup \{\sigma\}$ \Comment{Add new counterexample}
                    \State \textbf{continue}
                \EndIf
                \If{$\neg$ \textsc{IsFeasible}($\prog', \prog, \cexs$)} \Comment{Check feasibility of partial program}
                 \State {\bf continue}    
                  \EndIf
                  \State $N \gets \mathsf{ChooseNonterminal}(\prog')$
                  \For{$r \in \mathsf{Productions}(N)$ }
                   \State $\worklist \gets \worklist \cup \{ \mathsf{Expand}(\prog', r) \}$ \Comment{ Add expansions of $\prog'$ to worklist }
                   \EndFor

            \EndWhile

            \vspace{0.01in}
            \State \Return $\bot$
    \EndProcedure
    \end{algorithmic}
\end{algorithm}
\vspace{-0.3in}
\caption{Top-level transpilation Algorithm}  
\label{fig:synthalg}
\vspace{-0.2in}
\end{figure}

In this section, we discuss our transpilation algorithm based on neural-guided inductive program synthesis. We first describe our high-level approach  and then explain the key pruning that exploits the trace-compatibility assumption. 

\subsection{Top-level Algorithm}\label{sec:top-level}

Figure~\ref{fig:synthalg} summarizes our transpilation technique based on inductive program synthesis. The {\sc Transpile} procedure takes as input (i) the source program $\prog$,  (ii) a trained neural model $\nn$  (described in the next section), and (iii) a context-free grammar describing the target language. The output of {\sc Transpile} is either a program $\prog'$ that is equivalent to the source program $\prog$ or $\bot$ if no equivalent program is found.  At a high-level, the {\sc Transpile} algorithm is an instantiation of the counterexample-guided inductive synthesis paradigm but  (1) uses a new neural architecture to guide its search, and (2) prunes the search space by using program analysis techniques that exploit the trace-compatibility assumption.

Internally, {\sc Transpile} maintains a worklist $\worklist$ of partial programs and, as standard in CEGIS, a set of counterexamples $\cexs$. After initializing the worklist to an empty partial program (line 5), the algorithm enters  a loop in which it dequeues a partial program from $\worklist$ and expands it using one of the productions in $\grammar$. Specifically, at line 8, it invokes a procedure called {\sc ChooseBest} that returns the ``best" partial program according to the neural model $\nn$ as follows: 
\begin{equation}
\label{eq:nnbestprog}
    \argmax_{\prog' \in \worklist} \nn(\prog'  \ | \ \prog)
\end{equation}

If the chosen program $\prog'$ is complete (i.e., it has no non-terminals), then {\sc Transpile} invokes an equivalence checking engine to test whether $\prog$ and $\prog'$ are equivalent (line 10) and returns $\prog'$ if they are (line 12). In particular, the procedure $\mathsf{IsEquivalent}$ takes as input two programs \new{and the counterexample set} and returns a tuple $(\nu, \sigma)$ where $\nu$ is a boolean indicating whether the two programs are equivalent and $\sigma$ is a counterexample if they are not. Since our approach is based on CEGIS, the returned counterexample $\sigma$ is added to counterexample set $\cexs$ and the algorithm moves on to the next partial program in the worklist (lines 13-14). \new{Note that $\mathsf{IsEquivalent}$ first checks equality on the examples in the counterexample set before invoking a stronger verifier -- if one of these checks fails, the returned counterexample $\sigma$ is simply the counterexample from $\cexs$ which caused the failed check.}

On the other hand, if  $\prog'$ is a partial program with remaining non-terminals, our algorithm checks whether $\prog'$ is feasible under the trace compatibility assumption. In particular, the procedure {\sc IsFesible} (discussed in detail in the next subsection) takes as input $\prog, \prog'$ as well as the counterexamples $\cexs$ and checks whether they are trace-compatible on inputs $\cexs$ (line 15).  If they are not, the algorithm prunes the search space by not considering expansions of $\prog'$. Otherwise, it chooses a non-terminal $N$ used in $\prog'$ and adds all expansions of $\prog'$ obtained by replacing $N$ with some $w \in (\Sigma \cup V)^*$ (for $N \rightarrow w \in R) $ in $\prog'$ to the worklist.

\begin{figure}
\begin{algorithm}[H]
    \begin{algorithmic}[1]
        \Procedure{IsFeasible}{$\prog, \prog', \cexs$}
            \State \textbf{input:} Source program $\prog$
            \State \textbf{input:} Candidate program $\prog'$
            \State \textbf{input:} Counterexamples $\cexs$
            \vspace{0.05in}
            \For{$\sigma \in \cexs$}
                \If{$\nconjudge{\sigma}{\prog'}{\prog}$}
                    \State \Return $\emph{false}$
                \EndIf
            \EndFor
            \vspace{0.01in}
            \State \Return $\emph{true}$
    \EndProcedure
    \end{algorithmic}

\end{algorithm}
\vspace{-0.3in}
\caption{Algorithm for checking trace compatibility}  
\label{fig:prunealg}
\vspace{-0.2in}
\end{figure}

\subsection{Checking Trace Compatibility}
\label{sec:checkTraceCompat}

As illustrated by the discussion in Section~\ref{sec:top-level}, one of the novel aspects of our approach is its ability to rule out partial programs based on the trace-compatibility assumption. In this subsection, we discuss the {\sc IsFeasible} procedure, summarized in Figure~\ref{fig:prunealg}, for checking whether a partial program $\prog'$ has any completion that is trace-compatible with $\prog$ on inputs $\cexs$.

In more detail, {\sc IsFeasible} iterates over all inputs $\sigma \in \cexs$ and checks whether the partial program $\prog'$ satisfies the trace compatibility judgment $
\conjudge{\sigma}{\prog'}{\prog}
$
which indicates that $\prog'$ \emph{may} have a completion  that is trace compatible with $\prog$ on $\cexs$. More importantly, if $\nconjudge{\sigma}{\prog'}{\prog}$, this means that \emph{no} completion of $\prog'$  is trace compatible with $\prog$ and $\prog'$ can be pruned from the search space. Thus, the {\sc IsFeasible} procedure returns false if it finds any input $\sigma$ for which $\nconjudge{\sigma}{\prog'}{\prog}$ (line 7).

In the remainder of this subsection, we discuss our implementation of the $\conjudge{\sigma}{\prog'}{\prog}$ judgment. At a high-level, our technique for checking trace compatibility need to achieve two important goals:

\begin{itemize}[leftmargin=*]
    \item {\bf Pruning power:} In order to have good pruning power, our rules for checking trace compatibility need to be sufficiently precise. That is, if there is no completion of $\prog'$ that is trace compatible with $\prog$ on input $\sigma$, our {inference} rules should derive $\nconjudge{\sigma}{\prog'}{\prog}$ most of the time. 
\item {\bf Low overhead:} Since our transpilation algorithm invokes the trace compatability checker \emph{many} times, it is crucial that this procedure has low overhead. Thus, it must be efficient to check  whether $\nconjudge{\sigma}{\prog'}{\prog}$.
    \end{itemize}

With these goals in mind, we discuss our trace compatibility checking procedure in two steps. First, we describe inference rules that achieve the first goal (i.e., they are very precise and conceptually easy to understand); however, they fall short of our second goal. To address this shortcoming, 
we next describe more complex \emph{bidirectional} rules that achieve better scalability without sacrificing precision.


\subsubsection{Uni-directional  Rules for Trace Compatibility}\label{sec:forward}
\label{sec:unidirecrules}
Given a partial program $\prog'$ and an input valuation $\sigma$, our  rules over-approximate the possible outputs of $\prog'$ using a form of \emph{concolic execution}~\cite{sen2005cute}. The basic idea is to  introduce \emph{symbolic variables} to represent the unknown value of unshared non-terminals. On the other hand, for shared non-terminals, we know which \emph{concrete values} they can take when executing $\prog'$ on $\sigma$. We then propagate this mix of concrete and symbolic values using concolic execution and then check (using an SMT solver) whether it is possible for $\prog'$ to produce the output of $\prog$ on $\sigma$. If not, $\prog'$ and $\prog$ are guaranteed \emph{not} to be equivalent under the trace compatibility assumption, so we can safely prune $\prog'$ from the search space.

\begin{figure}
    \centering
    \textbf{Uni-directional Trace-Compatibility Checking Rules $\uparrow$} \\
    \begin{mathpar}
    \inferrule*[Right=$\nconsist$-$\uparrow$, width=10cm]{
       \sigma \not \vdash \prog' \consist \prog
    } {
       \nconjudge{\sigma}{\prog'}{\prog}
    }
    
    \inferrule*[Right=$\consist$-$\uparrow$, width=10cm]{
      \upjudge{\prog, \sigma}{\prog'}{\symexp} \\ \mathsf{SAT}(\symexp = \exec{\prog}{\sigma}) 
    } {
      \conjudge{\sigma}{\prog'}{\prog}
    }
    \\
    \inferrule*[Right=NTerm-$\uparrow$, width=10cm]{
       \neg \mathrm{Shared}(N) \\ \mathrm{FreshVar}(\nu)
    } {
       \upjudge{\prog, \sigma}{N}{\nu}
    }
    \\
    \inferrule*[Right=S-NTerm-$\uparrow$, width=10cm]{
       \mathrm{Shared}(N)  \\ 
       N \derives w \\ 
       w \in \prog \\
       c \in \coll{w}{\prog}{\sigma}
    } {
       \upjudge{\prog, \sigma}{N}{c}
    }
    \\
    \inferrule*[Right=VarTerm-$\uparrow$, width=10cm]{
       \mathrm{Variable}(v) \\ 
       \sigma[v] = \symexp
    } {
       \upjudge{\prog, \sigma}{v}{\symexp}
    }
    \\
    \inferrule*[Right=FirstOrderTerm-$\uparrow$, width=10cm]{
       \mathrm{Function}(f) \\ 
       \mathrm{FirstOrder}(f) \\
       \forall i. \, \upjudge{\prog, \sigma}{\pexp_i}{\symexp_i} \\ 
       \symexec{f}(\symexp_1, ..., \symexp_n) = \symexp
    } {
       \upjudge{\prog, \sigma}{f(\pexp_1, ..., \pexp_n)}{\symexp}
    }
    \\
    \inferrule*[Right=HigherOrderTerm-$\uparrow$, width=10cm]{
       \mathrm{Function}(f) \\
       \mathrm{HigherOrder}(f) \\
       \upjudge{\prog, \sigma}{\pexp_1}{[\symexp_1, ..., \symexp_k]} \\ 
       \upjudge{\prog, \sigma[i \mapsto \symexp_i]}{\pexp_2}{\symexp_i'} \\
       \symexec{f}([\symexp_1, ..., \symexp_k], \{ \symexp_1 \mapsto \symexp_1', ..., \symexp_k \mapsto \symexp_k' \}) = \symexp
    } {
       \upjudge{\prog, \sigma}{f(\pexp_1, \, \lambda i. \, \pexp_2)}{\symexp}
    }
    
    \end{mathpar}
    \caption{Uni-directional rules for checking trace compatibility between partial program in target language and a program in the source language. We assume that lambda bindings in higher-order functions do not shadow input variables. }
    \label{fig:forwardrules}
    \vspace{-0.1in}
\end{figure}

{To describe our pruning rules, we assume that the target language $\mathcal{L}_t$  is an instance of the following meta-grammar that takes as arguments (1) $F_H$, a set of higher-order functions, (2) $F$, a set of first-order functions, (3) $V$, a set of variables, and (4) $C$, a set of constants:}
\[
E \rightarrow F_H(E_1, \lambda V. \ E_2) \ | \ F(E_1, \ldots, E_N) \ | \ V \ | \ C
\]
We also assume that we have access to the abstract semantics~\cite{cousot1992abstract} for each function (including both first-order and higher-order) in the target language. Given some construct $f$, we write $\symexec{f}$ to denote the abstract semantics of $f$.

With these assumptions in place, we can now explain the inference rules from Figure~\ref{fig:forwardrules}. All rules with the exception of the first two derive judgments of the form: 
\[
\prog, \sigma \vdash \prog' \uparrow \psi
\]
where $\psi$ is a symbolic expression (i.e., an SMT term\footnote{The exact logical theory used in the SMT encoding depends on the abstract semantics, so we don't specify what logical theory these symbolic expressions belong to.}). 
The meaning of this judgment is that, given source program $\prog$ and input valuation $\sigma$, partial program $\prog'$ can produce symbolic expression $\psi$ under the trace compatibility assumption. Thus, according to the second rule in Figure~\ref{fig:forwardrules}, $\prog$ and $\prog'$ are trace compatible on valuation $\sigma$ if $\prog'$ can produce an expression $\psi$ that is consistent with the output of $\prog$ on $\sigma$ (i.e., we have $\mathsf{SAT}(\psi = \exec{\prog}{\sigma})$). Furthermore, according to the first rule, $\prog$ and $\prog'$ are \emph{not} trace compatible under $\sigma$ if we cannot derive $\sigma \vdash \prog \consist \prog'$ using the second rule.

Next, we explain the rules in Figure~\ref{fig:forwardrules} that derive judgments of the form $\prog, \sigma \vdash \prog' \uparrow \psi$. The rule labeled \textsc{NTerm-}$\uparrow$ is for unshared non-terminals and introduces a fresh variable $\nu$ to represent the unknown value of expressions derived for non-terminal $N$. The  \textsc{S-NTerm-}$\uparrow$ rule  is for \emph{shared} non-terminals $N$ and leverages the trace compatibility assumption. In particular,  under this  assumption, $N$ can only take on values observed when executing some term $w$ (for $N \derives w$) within $\prog$ on input valuation $\sigma$. Thus, we ``look up" the values that $w$ can take in $\prog$ and conclude that $N$ can take value $c$ only if $w$ can evaluate to $c$ when executing $\prog$ on input $\sigma$. 


The last three rules in Figure~\ref{fig:forwardrules} deal with terminal symbols in the grammar. The rule labeled \textsc{VarTerm-}$\uparrow$ is for variables and states that variable $v$ evaluates to $\sigma[v]$. The next two rules are for functions but differentiate between first-order and higher-order combinators.\footnote{We view constants as nullary functions; so there is no special rule for constants.} For first order functions, we first evaluate the arguments $\pexp_1, \ldots, \pexp_n$  through recursive invocation of the inference rules; let $\symexp_1, \ldots, \symexp_n$ be the  evaluation result. We then symbolically execute function $f$ on these arguments to obtain the symbolic expression $\symexp$ for $f(\pexp_1, \ldots, \pexp_n)$. 

The final rule is for higher-order functions, which, for simplicity, we assume to be of arity two, with the first argument being a list and the second argument being a lambda abstraction. Here, we first evaluate the first argument $\pexp_1$; suppose that the result is a list of length $k$ with symbolic elements $\symexp_1, \ldots, \symexp_k$. When evaluating the second argument $\lambda i. \pexp_2$, we first bind $i$ to $\symexp_i$ and then evaluate $\pexp_2$ as $\symexpp_i$. Finally, since $f$ takes  $[\symexp_1, \ldots, \symexp_n]$ as its first input and the function $\{ \symexp_1 \mapsto \symexp_1', \ldots, \symexp_k \mapsto \symexp_k'\}$ as its second input, we symbolically evaluate $f$ on these arguments to obtain the final result $\symexp$.

\begin{example}
\label{ex:forward}
Consider the following imperative program $\prog_s$, which takes in two lists of integers, $x_1$ and $x_2$, and produces an output list of all integers $i_1$ in $x_1$ for which $i_1*i_2+1$ is prime for some $i_2$ in $x_2$. Suppose that the synthesizer proposes the incorrect partial functional program $\prog_t$, where $I$ is a shared non-terminal representing integer expressions.

\begin{center}
    \begin{tabular}{cc}
Source program $\prog_s$: & Target partial program $\prog_t$: \\
\begin{minipage}{0.4\linewidth}
\begin{lstlisting}[numbers=none]
r = []
for(int $i_1$ : $x_1$) :
   for(int $i_2$ : $x_2$) :
      if prime($i_1*i_2+1$) :
         r.add($i_1$)  
return r
\end{lstlisting}
\end{minipage}
             &
\begin{minipage}{0.35\linewidth}

\begin{lstlisting}[numbers=none]
map($L_1$, $\lambda$ $i_1$. $I_1*(I_2+I_3$)))
\end{lstlisting}
\end{minipage}
    \end{tabular}
\end{center}
For input $\sigma = \{ x_1 \mapsto [1,2,3,4,5], x_2 \mapsto [11,70,61,72,61]\}$, $\prog_s$ produces the output $[1,1,2,3,4]$, and we can prune $\prog_t$ using these input-output examples and the trace compatibility assumption. In particular, each occurrence of shared non-terminal $I$ can take values of $i_1, i_2, i_1*i_2$, and $i_1*i_2+1$ in the source program. Since these expressions only take positive values, we infer that $I_1*(I_2+I_3)$ cannot be $1$. However, using the semantics of {\tt map}, this would imply that the output list cannot contain $1$, thereby contradicting the fact that the output list is $[1,1,2,3,4]$.

\end{example}

\subsubsection{Bidirectional Rules for Checking Trace Compatibility.} 
\label{sec:bidirecrules}

While the inference rules from Section~\ref{sec:forward} are sufficient for precisely checking trace compatibility, they can be very inefficient to execute. To gain some intuition, consider a term $f(\pexp_1, \ldots, \pexp_k)$ and suppose that each $\pexp_i$ can evaluate to $n$ different symbolic expressions. Then, there are a total of $n^k$ different expressions that $f$ can evaluate to. Thus, as this example illustrates, checking trace compatibility can be exponential in the size of the partial program being evaluated. 

To mitigate this problem, we augment our trace compatibility checking rules using backwards reasoning. The idea is to use the inverse semantics of constructs in the target language to obtain a specification $\varphi$ of each sub-term being evaluated. Then, if a sub-term evaluates to an expression $\symexp$ that is not consistent with $\varphi$ (i.e., $\mathsf{Unsat}(\varphi(\symexp)$), we know that $\symexp$ is infeasible and we need not propagate it forwards. Hence, the use of backwards reasoning allows us to control the number of symbolic expressions being propagated forwards and prevents this exponential blow-up in practice.

\begin{figure}
    \centering
    \textbf{Bidirectional Trace-Compatibility Checking Rules $\updownarrow$} \\
    \begin{mathpar}
    \inferrule*[Right=$\nconsist$-$\updownarrow$, width=10cm]{
       \sigma \not \vdash \prog' \consist \prog
    } {
       \nconjudge{\sigma}{\prog'}{\prog}
    }
    
    \inferrule*[Right=$\consist$-$\updownarrow$, width=10cm]{
       \udjudge{\prog, \sigma\textcolor{blue}{, y = \exec{\prog}{\sigma}}}{\prog'}{\symexp} \\
       \mathsf{SAT}(\symexp = \exec{\prog}{\sigma})
    } {
       \conjudge{\sigma}{\prog'}{\prog}
    }
    \\
    \inferrule*[Right=NTerm-$\updownarrow$, width=10cm]{
       \neg \mathrm{Shared}(N) \\ 
       \mathrm{FreshVar}(\nu) \\ 
       \textcolor{blue}{\mathsf{SAT}(\varphi[\nu])}
    } {
       \udjudge{\prog, \sigma\textcolor{blue}{, \varphi}}{N}{\nu}
    }
    \\
    \inferrule*[Right=S-NTerm-$\updownarrow$, width=10cm]{
       \mathrm{Shared}(N)  \\  
       N \derives w \\
       w \in \prog \\
       c \in \coll{w}{\prog}{\sigma} \\
       \textcolor{blue}{\mathsf{SAT}(\varphi[c])}
    } {
       \udjudge{\prog, \sigma\textcolor{blue}{, \varphi}}{N}{c}
    }
    \\
    \inferrule*[Right=VarTerm-$\updownarrow$, width=10cm]{
       \mathrm{Variable}(v) \\ 
       \sigma[v] = \symexp \\ 
       \textcolor{blue}{\mathsf{SAT}(\varphi[\symexp])}    
    } {
       \udjudge{\prog, \sigma\textcolor{blue}{, \varphi}}{v}{\symexp}
    }
    \\
    \inferrule*[Right=FirstOrderTerm-$\updownarrow$, width=10cm]{
       \mathrm{Function}(f) \\ 
       \mathrm{FirstOrder}(f) \\
       \forall i. \, \udjudge{\prog, \sigma\textcolor{blue}{, \symexecback{f}{-i}(\varphi, \symexp_1, ..., \symexp_{i-1})}}{\pexp_i}{\symexp_i} \\
       \symexec{f}(\symexp_1, ..., \symexp_n) = \symexp \\ 
       \textcolor{blue}{\mathsf{SAT}(\varphi[\symexp])}
    } {
       \udjudge{\prog, \sigma\textcolor{blue}{, \varphi}}{f(\pexp_1, ..., \pexp_n)}{\symexp}
    }
    \\
    \inferrule*[Right=HigherOrderTerm-$\updownarrow$, width=8cm]{
       \mathrm{Function}(f) \\ 
       \mathrm{HigherOrder}(f) \\
       \udjudge{\prog, \sigma\textcolor{blue}{, \symexecback{f}{-1}(\varphi)}}{\pexp_1}{[\symexp_1, ..., \symexp_k]} \\
       \udjudge{\prog, \sigma[i \mapsto \symexp_i]\textcolor{blue}{, \symexecback{f}{-2}(\varphi, \symexp_i)}}{\pexp_2}{\symexp_i'} \\
       \symexec{f}([\symexp_1, ..., \symexp_k], \{ \symexp_1 \mapsto \symexp_1', ..., \symexp_k \mapsto \symexp_k' \}) = \symexp \\
       \textcolor{blue}{\mathsf{SAT}(\varphi[\symexp])}
    } {
       \udjudge{\prog, \sigma\textcolor{blue}{, \varphi}}{f(\pexp_1, \, \lambda i. \, \pexp_2)}{\symexp}    
    }    
    
    \end{mathpar}
    \caption{Bidrectional Rules}
    \label{fig:bidirecrules}
    \vspace{-0.2in}
\end{figure}

With this intuition in mind, we now explain  the bidirectional trace compatibility checking rules shown in Figure~\ref{fig:bidirecrules} which derive judgments of the following form:
\[
\prog, \sigma, \varphi \vdash \prog' \updownarrow \symexp
\]
The meaning of this judgment is that $\prog'$ can evaluate to $\symexp$ under the assumption that (1) $\prog'$ satisfies $\varphi$ and (2) $\prog'$ is trace compatible with $\prog$ on valuation $\sigma$.  As expected, these rules are a refinement of those from Figure~\ref{fig:forwardrules} and the differences from Figure~\ref{fig:forwardrules} are indicated in blue. 
The rule $\consist$\textsc{-}$\updownarrow$ is the same as $\consist$\textsc{-}$\uparrow$, but it initializes the specification $\varphi$ of $\prog'$ to state that $\prog'$ must produce the same output as $\prog$ when executed on $\sigma$.
The rules labeled \textsc{NTerm-}$\updownarrow$, \textsc{S-NTerm-}$\updownarrow$, and \textsc{VarTerm-}$\updownarrow$ are exactly the same as the corresponding rules from Figure~\ref{fig:forwardrules} with the only difference being the consistency check between the specification  and the result of the concolic evaluation. 

The last two rules for function terms are slightly more involved and make use of the inverse semantics. In particular, in the rule labeled \textsc{FirstOrderTerm-}$\updownarrow$, we first compute the specification for argument $\pexp_i$ using the inverse semantics of $f$ with the respect to its $i$'th argument. Specifically, we use the notation $\symexecback{f}{-i}$ to denote the function that takes as input the specification for the whole term and the (symbolic) values for expressions $\symexp_1, \ldots, \symexp_{i-1}$  and yields the specification for the $i$'th argument. Then, when evaluating $\symexp_i$, we use $\symexecback{f}{-i}(\varphi, \symexp_1, \ldots, \symexp_{i-1})$ as the specification.

\begin{example}
Consider the term $+(N_1, N_2)$ and suppose that we have:
\[
\prog, \sigma, \symexecback{+}{-1}(y>0) \vdash N_1 \updownarrow 1
\]
Note, the argument $y > 0$ of $\symexecback{+}{-1}$ is the specification for the entire term, i.e., $+(N_1, N_2)>0$. Then, using the inverse semantics of $+$, we can infer the specification $y>-1$ for the second argument $N_2$. This specification can be used to prune all negative values for $N_2$. 
\end{example}

The next rule for higher-order combinators is similar and uses the inverse semantics in exactly the same way. In particular, note that the satisfiability check reduces the space of possible evaluation results $\psi$ for $f(e_1, \lambda i_1. e_2)$  and therefore reduces the number of terms propagated forward.

The following theorem states that our pruning rules only rule out partial programs that are not trace compatible with the source program on valuation $\sigma$.

\begin{theorem}
\old{Suppose that $\nconjudge{\sigma}{\prog}{\prog'}$. Then, there is no completion of $\prog'$ that is trace-compatible with $\prog$ on input $\sigma$.}
\new{Suppose that $\nconjudge{\sigma}{\prog}{\prog'}$ and both the forward semantics $\symexec{\cdot}$ and backward semantics $\symexecback{\cdot}{-i}$ provided are sound. Then, there is no completion $\prog_t$ of $\prog'$ that is trace-compatible with $\prog$ on input $\sigma$ where $\exec{\prog_t}{\sigma} = \exec{\prog}{\sigma}$.}
\end{theorem}

\ifExt
\begin{proof}
See Appendix~\ref{app:pruning}
\end{proof}
\else
\begin{proof}
Please see the \href{https://rebrand.ly/r1ckr0l13r}{extended version} for full proofs.
\end{proof}
\fi
\section{Cognate Grammar Network}
\label{sec:net}

The transpilation algorithm from Figure~\ref{fig:synthalg} uses a neural model to guide its search. 
In this section, we discuss the \emph{cognate grammar network} (\net) that serves as our custom neural model in this context. In what follows, we first explain the  \emph{abstract syntax network} (\asn) model on which \net is based \cite{asn}, then describe how \net extends the \asn architecture with a copying mechanism \cite{jia-liang-2016-data,see-etal-2017-get} to handle shared nonterminals, and, finally, we give a brief overview of our training phase.

\subsection{ASN Preliminaries }
\label{subsec:asn}

Our goal is to model the distribution $p(\prog_t \mid \prog_s)$ over programs $\prog_t$ in the target language given programs $\prog_s$ in the source language. An abstract syntax network models this distribution by decomposing the probability of a complete target program $\prog_t$ into a series of probabilities for each AST node in $\prog_t$.

\paragraph{Notation.} To formalize this decomposition, we first introduce some notation. Following \cite{ye2020optimal}, given node $n$ in (complete) program $\prog$ (simply denoted $n \in \prog$), we define an AST path $\astpath(\prog,n) = ((n_1, i_1), \ldots, (n_k, i_k))$ to be a sequence of (node, index) pairs where node $n_{j+1}$ is the $i_j$'th child of node $n_j$, and the $i_k$'th child of node $n_k$ is $n$. In other words,  $\astpath(\prog,n)$ is the path from the root node of the AST to node $n$. For each node $n$, we use the notation $\derive{n}$ to denote the CFG production used to assign a terminal symbol to $n$.

\begin{example}
The AST for the simple functional program $\mathsf{map}(x, \lambda i. i + 1)$ is pictured below. The name for each node is given as a subscript.
\begin{center}
    \includegraphics[scale=.2]{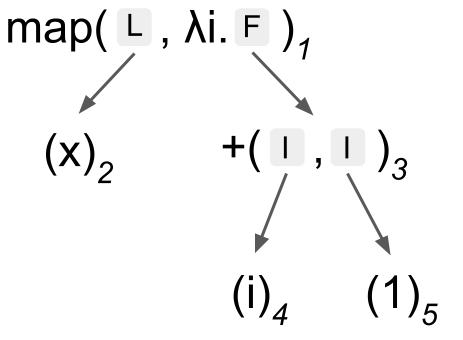}
\end{center}
The ast path to node $n_4$, denoted $\astpath(\prog, n_4)$, is $((n_1,2),(n_3,1))$ and the production rule used to assign $n_4$, $\derive{n_4}$, is $I \rightarrow i$.
\end{example}

The \asn $\nn$ models $p(\prog_t \mid \prog_s)$ using the distribution $\nn(\derive{n} \mid \astpath(\prog_t, n), \prog_s)$ over grammar productions $\derive{n}$ for each node $n$ in $\prog_t$ given the AST path $\pi$ to $n$ and the source program $\prog_s$. In particular, the probability $\nn(\prog_t \mid \prog_s)$ is computed 
as follows:
\begin{equation}
\label{eq:nnbestprod}
    \nn(\prog_t \mid \prog_s) = \prod_{n \in \prog_t} \nn(\derive{n} \mid \astpath(\prog_t, n), \prog_s)
\end{equation}
In other words, the probability of program $\prog_t$ given $\prog_s$ is given as the product of the probability of each grammar rule used to construct $\prog_t$.

Following past work \cite{ye2020optimal}, we compute Equation~\ref{eq:nnbestprod} by first encoding the source program $\prog_s$ with a $\bilstm$ \cite{lstm} over its linearized representation $\linear(\prog_s)$ to form an embedding $h_p = \bilstm(\linear(\prog_s))$. In this work, $\linear(\prog_s)$ simply treats $\prog_s$ as a sequence of tokens, where each token corresponds to a keyword or symbol from $\prog_s$ (e.g., {\tt for}, {\tt if}, {\tt =}, etc.).
We then use an LSTM\footnote{Note that this is a unidirectional LSTM, since we are encoding a sequential ``history'' of the node to predict its children. The most recent ancestors in the tree are what matter the most.} to produce embedding $h_n = \lstm(h_p, \astpath(\prog_t, n))$ for node $n$ of the target program AST, which will be used to form the distribution over production rules at this node. In particular, given embedding $h_n$, the probability for each production rule at node $n$ is computed using a feedforward neural network ($\ffnn$) module with attention over the input source program:
\begin{equation}
\label{eq:asn}
    \nn(\cdot \mid \astpath(\prog_t, n), \prog_s) = \softmax(\ffnn(h_n;\attn(h_n, \bilstm(\linear(\prog_s)))))
\end{equation}
Attention is a standard tool in sequence-to-sequence models for making information from the input more accessible to the decoder model \cite{bahdanau2015neural, luong-etal-2015-effective}. In our case, the attention layer computes a distribution over tokens in the linearized source program $\linear(\prog_s)$ for each node $n$, indicating which input tokens contribute most (and least) to the probability of production rules used at $n$. 

\subsection{Proposed CGN Model}

While the \asn model allows us to model $p(\prog_t \mid \prog_s)$, it is not necessarily well suited to our problem setting.  In particular,  source and target programs in our context are generated using cognate grammars, where some terms are shared between the two programs. For nodes corresponding to shared non-terminals, we want our model to produce a probability distribution over \emph{terms} in the source program as opposed to \emph{production rules} in the grammar. Thus, we propose the \emph{cognate grammar network (CGN)} model, which comprises (1) a \emph{copy} component for handling nodes  that correspond to shared non-terminals,  and (2) an \asn component for handling nodes corresponding to unshared non-terminals. In essence, the copy component directly \emph{copies} terms from the source program to the target program \new{in the same way that pointer networks copy portions of input text during text summarization \cite{see-etal-2017-get}}, while the \asn component builds target program terms from grammar rules in the target language.


In more detail, for each node $n$ that \emph{does not} correspond to a shared non-terminal (i.e., the LHS of $\derive{n}$ is \emph{not} a shared non-terminal), the ASN component of the \net simply uses Equation~\ref{eq:asn} for determining the probability of a production rule at $n$. However, for each node $n$ which corresponds to a shared non-terminal $N$, the copy component of the \net computes the distribution $p(t \mid \astpath(\prog_t,n), \prog_s)$ over terms $t$ from the source program derivable from $N$ ($N \derives t)$ given the AST path to $n$ and source program $\prog_s$. To compute this, we must first compute an encoding for the entire source program term $t$, which (potentially) consists of multiple tokens, each with their own separate $\lstm$ hidden state embedding. Following work on encoding spans in constituency parsing models \cite{kitaev-klein-2018-constituency}, we compute the encoding of $t$ as the difference of the LSTM hidden states corresponding to the start and end tokens of $t$ as follows:
\begin{equation}
    \encode(t) = \frac{1}{2}(\hidden_S - \hidden_E)
\end{equation}
where $\hidden_S$ and $\hidden_E$ are the LSTM hidden states for the start and end tokens respectively, and $\frac{1}{2}$ is a scaling factor.

To compute the probability of copying some source program term $t$, we learn a weight matrix $W$ which is combined with the term encoding $\encode(t)$ and target program node embedding $h_n$ as follows:
\begin{equation}
\label{eq:copy}
    \nn(t \mid \astpath(\prog_t, n), \prog_s) = \softmax(h_n^\top W \encode(t))
\end{equation}

\subsection{Training}

Our model is trained on a labeled dataset of paired programs $\mathcal{D} = (\prog_s, \prog_t)$. As is standard for this type of model, we maximize the conditional log probability of the target program given the source program:
\begin{equation}
\mathcal{L}(\theta) = \sum_{i=1}^{|\mathcal{D}|} \log \nn(\prog_t \mid \prog_s)
\end{equation}
It should be noted that we can directly minimize this quantity, as the copy component is invoked if and only if a node corresponds to a shared non-terminal and the corresponding source-side shared terminal is unambiguous. In other words, the ground-truth decision for each production rule choice $\derive{n}$ or term choice $t$ at node $n$ is unambiguous, regardless of whether $n$ corresponds to a shared or unshared non-terminal. 

{Similar to prior work \cite{ye2020optimal}, we train the networks using \emph{complete} programs as opposed to the partial programs seen at test time. This is possible because both the \asn and \net compute a probability distribution based on the source program $\prog_s$ and AST path $\pi(\prog_t, n)$ to the node $n$ being expanded (Equations~\ref{eq:asn} and \ref{eq:copy}). Thus, when learning these probability distributions from complete program pairs $(\prog_s, \prog_t)$, we train the network for each node $n$ in $\prog_t$ where the ground-truth label $\derive{n}$ (or $t$ for \net) is given by the choice of grammar rule (or term for \net) used to expand $n$ in $\prog_t$ and the embeddings for $\pi(\prog_t, n)$ and $\prog_s$ are computed as described in Section~\ref{subsec:asn}.}
\section{Implementation and Instantiations}
\label{sec:implementation}

We implemented our transpilation algorithm in a tool called \tool, which is parameterized over the source and target language grammars. 
In this section, we first discuss  salient aspects of our implementation and then describe  how to instantiate \tool for a pair of source/target languages.

\subsection{\tool Implementation}
\tool is implemented in Python and performs several optimizations over the algorithm presented in Sections~\ref{sec:alg} and~\ref{sec:net}.


\paragraph{Training.} For training our \net, we use the Adam optimizer \cite{adam}, with early stopping based on the accuracy achieved on the validation set and a dropout mechanism~\cite{srivastava14dropout} to prevent model overfitting. Additionally, we apply the gradient norm clipping method \cite{pascanu13clipping}  to stabilize  training.

\paragraph{Training {data}.} As there are no existing datasets for our target application domains, we  generated synthetic source/target program pairs to train the \net. Our data generation procedure relies on a key insight: \textit{while translating imperative code to functional APIs is challenging, translating functional APIs to imperative code is straightforward}. In particular, this \emph{reverse} translation problem can be achieved with a small set of composable translation rules. We propose a translation function $\translate(e,r)$ that converts a functional expression $e$ into equivalent imperative code $c$ with return value saved to variable $r$. For example, consider the following implementation of $\translate$ for the functional operator $map$:
\begin{center}
\centering
  \begin{minipage}{\linewidth}
\begin{lstlisting}[linewidth=10.5cm,frame=none, language=python, numbers=none]
                               translate($e_1$, $r_1$)
                               for $i \in r_1$:    
translate(map($e_1$, $\lambda i$. $e_2$), $r$) =         translate($e_2$, $r_2$)
                                 $r$.add($r_2$)
\end{lstlisting}    
\end{minipage}
\end{center}
This translation rule simply inlines the result of translating subexpressions $e_1$ and $e_2$, and rules for other higher-order operators follow a similar pattern. Thus, to generate training data, we randomly sample functional programs and use the the translator described above to generate matching imperative programs.
\new{
Note that using this same approach in reverse is not feasible -- i.e., applying these translation rules in reverse would not yield a solution to our problem. This is because these rules construct a \emph{separate loop} for each higher-order operator in a functional expression. Thus, to apply the rules in reverse (i.e., to apply compositional rules to imperative code), each loop would need to correspond to exactly one higher-order operator, which is rarely the case. For example, the three higher-order operators in Figure~\ref{fig:motivex_func} correspond to only two loops in Figure~\ref{fig:motivex_imp}.
}

\paragraph{Equivalence checker.} Our implementation utilizes a (bounded) verifier implemented on top of {\sc Rosette}~\cite{torlak2013growing} to check equivalence between the source and target programs. Specifically, the checker takes  as input two programs $\prog_s, \prog_t$ as well as a symbolic input $x$ and verification bound $K$ that determines how many times loops are unrolled. It then symbolically executes $\prog_t$ and $\prog_s$ on $x$, resulting in symbolic output states $y_t$ and $y_s$. If $y_t$ and $y_s$ are proven equal, $\prog_t$ and $\prog_s$ are known to be equivalent up to bound $K$. While the use of a bounded verifier can, in principle, result in \tool producing the wrong target program, we have not found this to be an issue in practice.


\subsection{Instantiating \tool}\label{sec:inst}

We have instantiated \tool  for two pairs of source/target languages. Our first client involves translating standard Java code to functional equivalents written using the Stream API, and the second requires converting imperative Python to code snippets using the {\tt functools} API. To give the reader a sense of what it takes to customize \tool, we describe how we instantiated it in the context of our Java client.

\begin{figure}[t]
  \centering
  \[
  \begin{array}{lll}
    \mathrm{Application} \ A \ & \rightarrow \mathsf{map}(A, \lambda V. A) \ | \ \mathsf{filter}(A, \lambda V. E) \ | \ \mathsf{flatmap}(A, \lambda V. A) \ | \\
    & \;\;\;\;\,\mathsf{find}(E, A, \lambda V. E) \ | \ \mathsf{fold}(E, A, \lambda V,V. E) \ | \ E \\
     \mathrm{Expression} \ E \ &  \rightarrow V \ | \ C \ | \ f(\overline{E})  \\
    \mathrm{Variable} \ V \ & \rightarrow x \ | \ i_1 \ | \ \ldots \ | \ i_m \\
    \mathrm{Constant} \ C \ & \rightarrow k \in \mathbb{Z} \ | \ b \in \mathbb{B}\\
  \end{array}
 \]
 \vspace{-0.1in}
  \caption{Target language \fdsl where $f$ denotes a first-order function}
  \label{fig:fdsl}
  \vspace{-0.15in}
\end{figure}



\paragraph{Grammar and shared non-terminals.} Since most readers are familiar with standard Java, we only present our target language, which is shown in Figure~\ref{fig:fdsl} and which we refer to as \fdsl in the rest of this section. \fdsl allows functional programs composed of common higher-order functional operators, such as $\mathsf{map}$, $\mathsf{filter}$, and $\mathsf{fold}$. In addition, \fdsl also contains two other operators $\mathsf{find}$ and $\mathsf{flatmap}$. As its name indicates, $\mathsf{flatmap}(e_1, \lambda i. e_2)$ is  similar to $\mathsf{map}$ in that it takes each element of $e_1$ and applies $e_2$ to it. However, for $\mathsf{flatmap}$, $e_2$ is a function which produces a list, and the final result is flattened to a single list. For example, $\mathsf{flatmap}(x, \ \lambda i. \ \mathrm{factors}(i))$ will produce a list of all factors of integers appearing in the list $x$. As for $\mathsf{find}(e_1, e_2, \lambda i. e_3)$, it returns the first element of  list $e_2$ satisfying condition $e_3$. If no such element is found, the default value $e_1$ is returned. For example, $\mathsf{find}(0, \ x, \ \lambda i. \ i \% 2 == 1)$ finds the first odd element of $x$, or $0$ if none is found.

In order to leverage the synthesis technique proposed in this paper, we formulate the source and target languages as \emph{cognate grammars} as defined in Def.~\ref{def:cognate}. Our choice of shared non-terminals is motivated by the observation from Section~\ref{sec:intro}: while the high-level structure of the source and target programs are very different, low-level expressions tend to be shared. In the grammar shown in Fig~\ref{fig:fdsl}, we designate  non-terminals $E$,  $V$, and $C$ as shared, while the top-level non-terminal $A$ is unshared.

\paragraph{Forward semantics.} Recall that our basic rules from Figure~\ref{fig:forwardrules} make use of the forward semantics of the 
target language. Figure~\ref{fig:forward-semantics} shows the semantics that we use for \fdsl in our implementation. 
Since these rules are based directly on the standard semantics of {\tt map}, {\tt fold} etc., we do not explain them in detail. Also,  these rules utilize standard auxiliary functions such as {\tt append} and can be easily converted into an SMT encoding.

\paragraph{Backward semantics.} Recall from Section~\ref{sec:bidirecrules} that our method also utilizes the backward semantics of each target language construct to make trace compatibility checking more practical.  Figure~\ref{fig:backward-semantics} shows the inverse semantics for \fdsl. Here, we use the notation $\symexecback{f}{-i}(\varphi, \psi_1, \ldots \psi_{i-1})$ to indicate the specification for the $i$'th argument of  $f$ given the specification $\varphi$ of the output and symbolic values for the first $i$ arguments. We further assume that all specifications $\varphi$ contain a single free variable $y$ that refers to the output of the expression. 


\begin{figure}[t]
  \centering
  \begin{tabular}{c}
           Forward Semantics $\symexec{\cdot}$ \\
           \\ 
  $\begin{array}{lll}
  \symexec{\mathrm{map}}([\symexp_1, \ldots, \symexp_k], \uplus \{ \symexp_i \mapsto \symexp_i^\prime \}) & \defeq & [\symexp_1^\prime, \ldots, \symexp_k^\prime] \\

  \symexec{\mathrm{flatmap}}([\symexp_1, \ldots, \symexp_k], \uplus \{\symexp_i \rightarrow [\symexp_{i1}^\prime, \ldots, \symexp_{il_i}^\prime]\}) & \defeq & [\symexp_{11}^\prime, \ldots, \symexp_{1l_1}^\prime, \ldots, \symexp_{k1}^\prime, \ldots, \symexp_{kl_k}^\prime] \\

   \symexec{\mathrm{filter}}([\symexp_1, \ldots, \symexp_k], \uplus \{\symexp_i \rightarrow \symexp_i^\prime\}) & \defeq & \mathrm{append}(\symexpbar_1, ..., \symexpbar_k) \\
  & & \ \ \ \ \text{where } \symexpbar_i = (\symexp_i^\prime \ ? \ \symexp_i : \bot) \\

 \symexec{\mathrm{find}}(\symexp_0, [\symexp_1, \ldots, \symexp_k], f) & \defeq & (\len(l) = 0 \ ? \ \symexp_0 : \mathrm{hd}(l)) \\
  & & \ \ \ \ \text{where } l = \symexec{\mathrm{filter}}([\symexp_1, \ldots, \symexp_k], f) \\

  \symexec{\mathrm{fold}}(\symexp_0, [], \{\}) & \defeq & \symexp_0 \\

  \symexec{\mathrm{fold}}(\symexp_0, [\symexp_1, \ldots, \symexp_k], \uplus \{(\symexp_{i-1}, \symexp_i) \rightarrow \symexp_i^\prime\}) & \defeq & \symexp_k^\prime

  \end{array}$ 
  \end{tabular}
 \vspace{-0.1in}
  \caption{Forward semantics for \fdsl.}
  \label{fig:forward-semantics}
  \vspace{-0.15in}
\end{figure}

\begin{figure}[t]
  \centering
  \begin{tabular}{c}
  Backward Semantics $\symexecback{\cdot}{-i}$ \\
  \\
  $\begin{array}{lll}
  \symexecback{\mathrm{map}}{-1}(y = [\symexpp_1, \ldots, \symexpp_n]) & \defeq & \len(y) = n \\
  
  \symexecback{\mathrm{map}}{-2}(y = [\symexpp_1, \ldots, \symexpp_n], \symexp_i) & \defeq & y = \symexpp_i \\

  \symexecback{\mathrm{flatmap}}{-1}(y = [\symexpp_1, \ldots, \symexpp_n]) & \defeq & \mathsf{true} \\

  \symexecback{\mathrm{flatmap}}{-2}(y = [\symexpp_1, \ldots, \symexpp_n], \symexp_i) & \defeq & \mathrm{subarr}(y, [\symexpp_1, \ldots, \symexpp_n]) \ \land \ \len(y) \leq n \\

  \symexecback{\mathrm{filter}}{-1}(y = [\symexpp_1, \ldots, \symexpp_n]) & \defeq & \len(y) \geq n \land \mathrm{subseq}([\symexpp_1, \ldots, \symexpp_n],y)\\

  \symexecback{\mathrm{filter}}{-2}(y = [\symexpp_1, \ldots, \symexpp_n], \symexp_i) & \defeq & \big ( \bigwedge_{j=1}^n \symexp_i \neq \symexp_j^\prime \big ) \rightarrow y = \mathsf{false} \\
  
  \symexecback{\mathrm{find}}{-1}(y = \symexpp) & \defeq & \mathsf{true} \\

  \symexecback{\mathrm{find}}{-2}(y = \symexpp, \symexp_0) & \defeq & (\symexp^\prime \neq \symexp_0) \rightarrow \symexp^\prime \in y \\
 
   \symexecback{\mathrm{find}}{-3}(y = \symexpp, \symexp_0, \symexp_i) & \defeq & \symexp^\prime \neq \symexp_i \rightarrow y = \mathsf{false} \\

   \symexecback{\mathrm{fold}}{-1}(y = \symexpp) & \defeq & \mathsf{true} \\
   
   \symexecback{\mathrm{fold}}{-2}(y = \symexpp, \symexp_0) & \defeq & \mathsf{true} \\
   
   \symexecback{\mathrm{fold}}{-3}(y = \symexpp, \symexp_0, \symexp_i) & \defeq & y = 
   \symexpp \ \ \text{if } \symexp_i \text{ final arg}
   
   \end{array}$
 \end{tabular}
 \vspace{-0.1in}
  \caption{Backward semantics for \fdsl.}
  \label{fig:backward-semantics}
  \vspace{-0.2in}
\end{figure}

\section{Experiments}
\label{sec:experiments}

In this section, we describe the results of a series of experiments that are designed to answer the following research questions:

\begin{itemize}
    \item [\textbf{RQ1}] How does \tool compare to existing  techniques?
    \item [\textbf{RQ2}] How important are the various design decisions underlying \tool? In particular, how important are our proposed neural architecture and pruning techniques?
    \item [\textbf{RQ3}] How does the trace compatibility assumption limit \tool's ability to transpile real-world benchmarks?
\end{itemize}

All experiments are run with a 5-minute timeout on a 2019 MacBook Pro with an 8-core i9 processor and 16 GB RAM. 

\subsection{Benchmark Collection}

As mentioned earlier in Section~\ref{sec:inst}, we instantiated \tool for two types of transpilation tasks, namely (1) converting imperative Java to functional code using the {\tt Stream} API and (2) translating Python  code to a functional variant using the {\tt functools} API.  Unlike our training benchmarks that were generated synthetically, it is important to evaluate \tool on  \emph{test} benchmarks that are representative of real-world code. Towards this goal, we implemented a web crawler to extract real-world transpilation tasks from Github and Stackoverflow. In particular, our script looks for Github commits and Stackoverflow posts that meet the following criteria:
\begin{enumerate}[leftmargin=*]
\item Contains relevant keywords such as \texttt{Java stream}, \texttt{Pythonic loops}, or \texttt{List comprehension}.
\item Includes both the imperative code and its functional translation in the versions before/after the commit or  Stackoverflow post/answer.
\item Passes our non-triviality checks such as  the imperative version containing at least one loop, and the  functional version containing at least two higher-order combinators. 
\end{enumerate}

Using this methodology, we obtained a total of $91$ benchmarks,  46 of which are for Java and 45 of which are Python. To give the reader some intuition about the complexity of these benchmarks, Table~\ref{tab:benchmark_info} presents various statistics about the source/target versions of these programs.


\begin{table}[t]
\centering
 \begin{tabular}{|c|c|c|c|c|c|c|} 
 \hline
 Client & \multicolumn{3}{c}{Imperative Stats} & \multicolumn{3}{|c|}{Functional Stats} \\
 \hline
 & \multicolumn{3}{c|}{AST Nodes} & 
 \multicolumn{3}{c|}{AST Nodes}
 \\
 & \multicolumn{1}{c}{Avg} & \multicolumn{1}{c}{Med} & Max & 
 \multicolumn{1}{c}{Avg} & \multicolumn{1}{c}{Med} & Max
 \\
 \cline{2-7}
 Stream & 66 & 62 & 127 & 
 40 & 35 & 100 
 \\
 Pythonic & 76 & 72 & 198 & 
 38 & 36 & 159 
 \\
 \hline
 \end{tabular}
 \caption{Benchmark Information.}
 \label{tab:benchmark_info}  
    \vspace{-0.3in}
\end{table}


\subsection{Comparison against Existing Techniques}

To answer our first research question, we compared the performance of \tool against existing techniques. In particular, we consider the following four baselines:
\begin{itemize}[leftmargin=*]
    \item {\bf CEGIS-Neo}, a CEGIS-based synthesizer that uses {\sc Neo}~\cite{neo} as the inductive synthesis backend and our equivalence checker as the verifier
    \item {\bf Enum-ASN}, an enumerative synthesizer  in which we enumerate 
 syntactically valid code according to its likelihood using an abstract syntax network (ASN) model~\cite{asn} and then verify correctness using our equivalence checker
 \item {\bf Enum-Seq2Seq}, another enumerative synthesizer based on a Sequence to Sequence (\seqtoseq)  model \cite{sutskever2014sequence} in which we enumerate programs based on the score of the \seqtoseq model and verify correctness using our equivalence checker
 \item {\bf Netbeans}, a translator (implemented as a Netbeans plugin) that uses a set of handcrafted rules to modernize  imperative Java code to use the Stream API~\cite{gyori2013crossing}
 \end{itemize}
 Note that the first three of these baselines are applicable to both of our clients, while the last one is specific to Java. {For all tools which rely on a neural component (i.e., \tool, \enumasn, and \enumseqtoseq) we use the same training data and training methodology described in Section~\ref{sec:implementation}. Additionally, to set hyperparameters for these systems, we experimented with $10$ different common/intuitive settings for each and used the one which produced the best results.}

\begin{figure}
    \centering
    \subfloat[Stream results.]{
    \includegraphics[width=.95\linewidth]{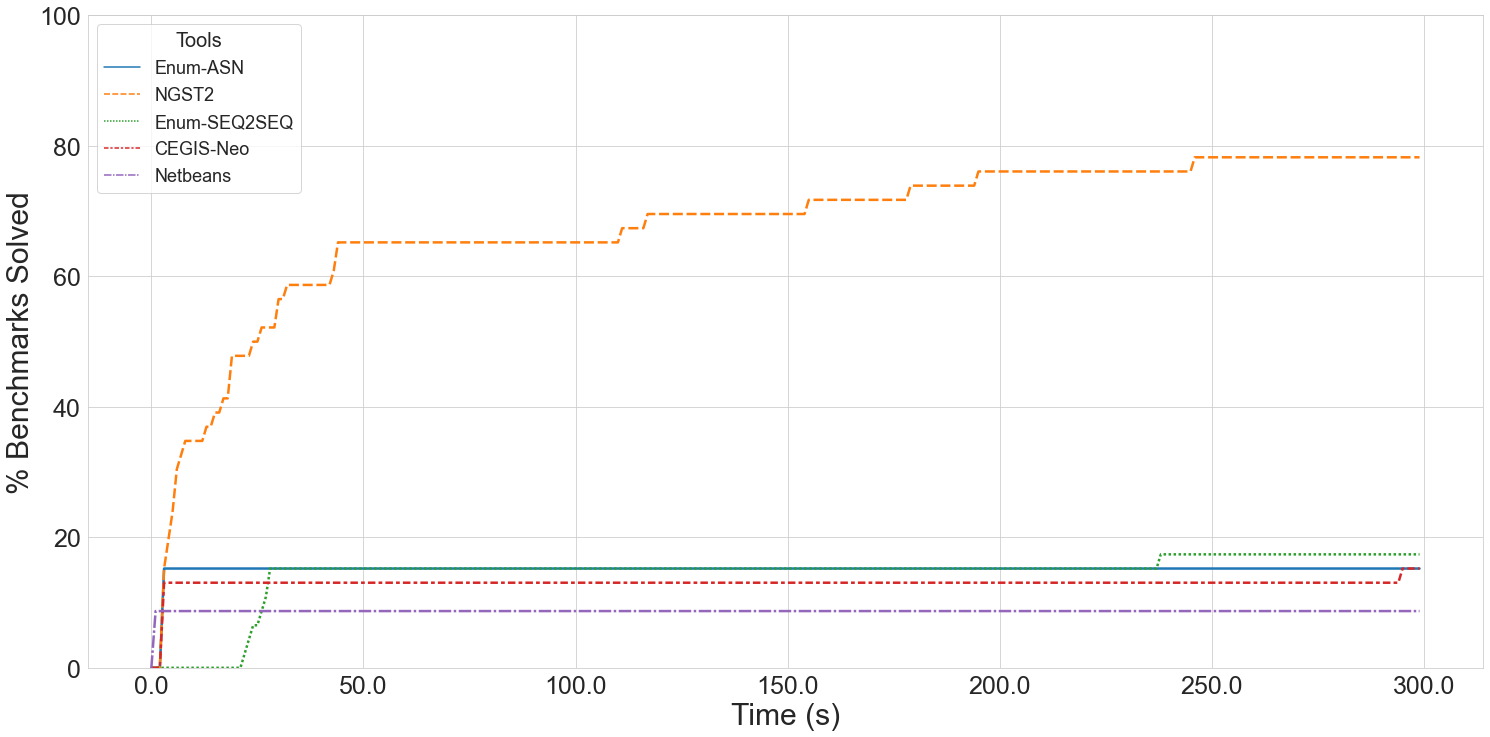}
    }
    \hfill
    \subfloat[Python results.]{
    \includegraphics[width=.95\linewidth]{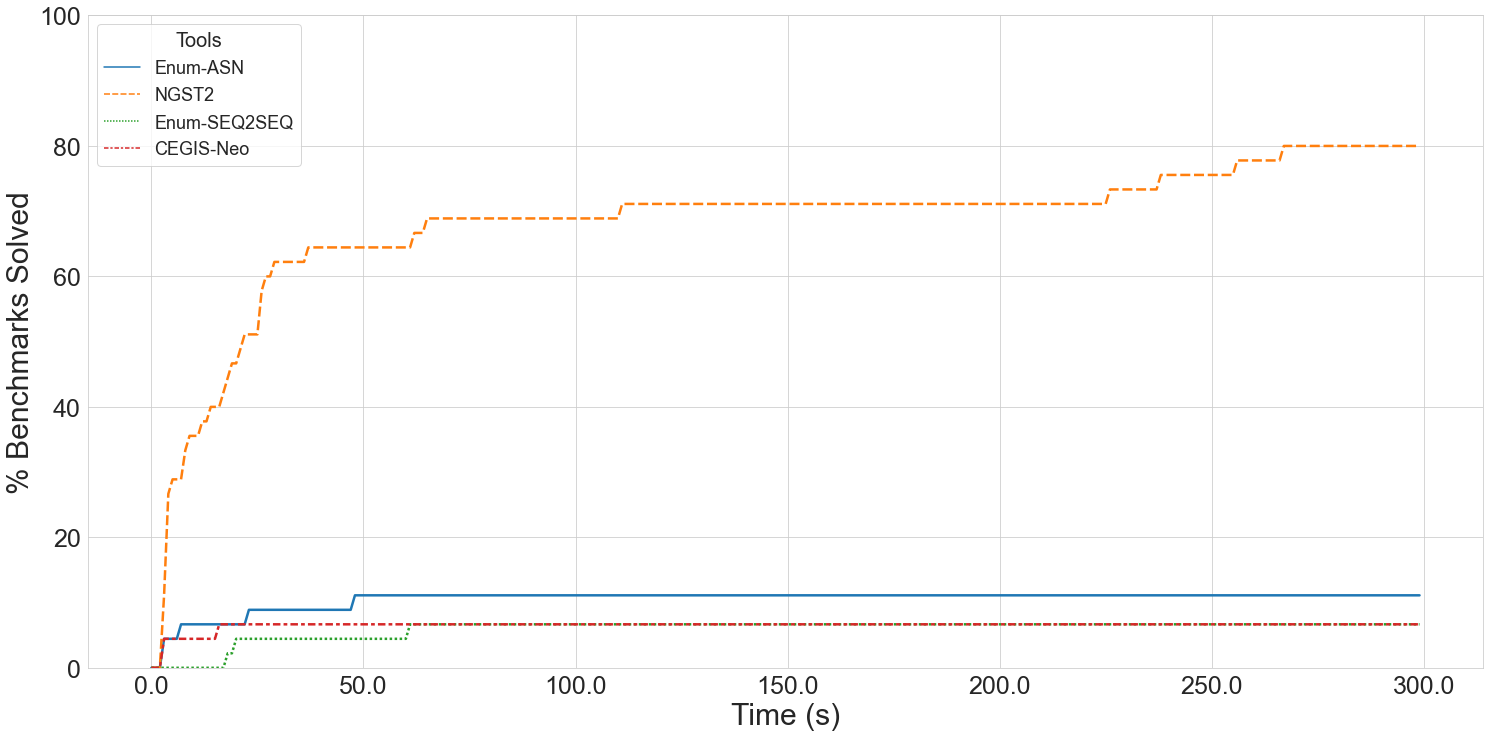}
    }
    \caption{Performance comparison with existing tools.}
    \label{fig:baselineGraphs}
\end{figure}

 The results of this experiment are summarized   in Figure~\ref{fig:baselineGraphs}. Here, the x-axis shows the   time (per benchmark) and the y-axis presents the percentage of benchmarks solved within that time limit.  As we can see in both figures, \tool significantly outperforms all baselines in both clients. In particular, \tool solves at least 60\% more benchmarks than any other baseline across both clients. We believe the major cause of poor performance for the other synthesizers (\enumasn, \enumseqtoseq, \cegisneo) is the large search space of programs to be searched. In contrast to these tools, \tool leverages the trace compatibility assumption to prune large parts of the search space. The poor performance of the \netbeans plugin is likely due to an inadequate set of hand-made translation templates; as shown, in the few cases where the plugin did have the correct templates, translation was very fast.
 

\vspace{4mm}
 \fbox{
\begin{minipage}{0.9\linewidth}
{\bf Result for RQ1:}
\tool significantly outperforms existing baselines for both the Java and Python clients. In particular, \tool solves 78\% of the Java benchmarks, while its closest competitor (\enumseqtoseq) solves 17\%. For the Python client, \tool solves 80\% of benchmarks, while the best baseline (\enumasn) solves only 11\%. 
\end{minipage}
}

\subsection{Ablation Studies}

\begin{figure}
    \centering
    \subfloat[Stream results.]{
    \includegraphics[width=.95\linewidth]{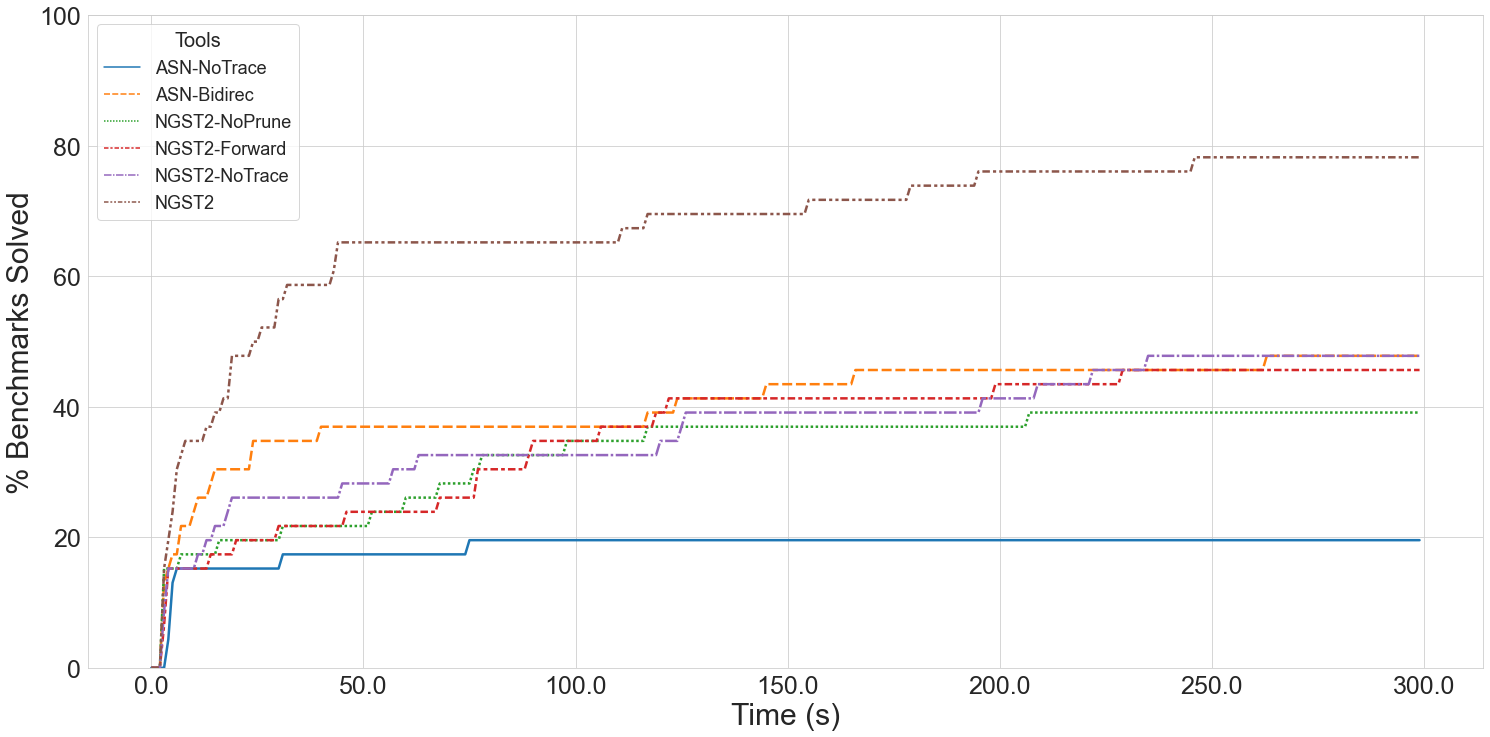}
    }
    \hfill
    \subfloat[Python results.]{
    \includegraphics[width=.95\linewidth]{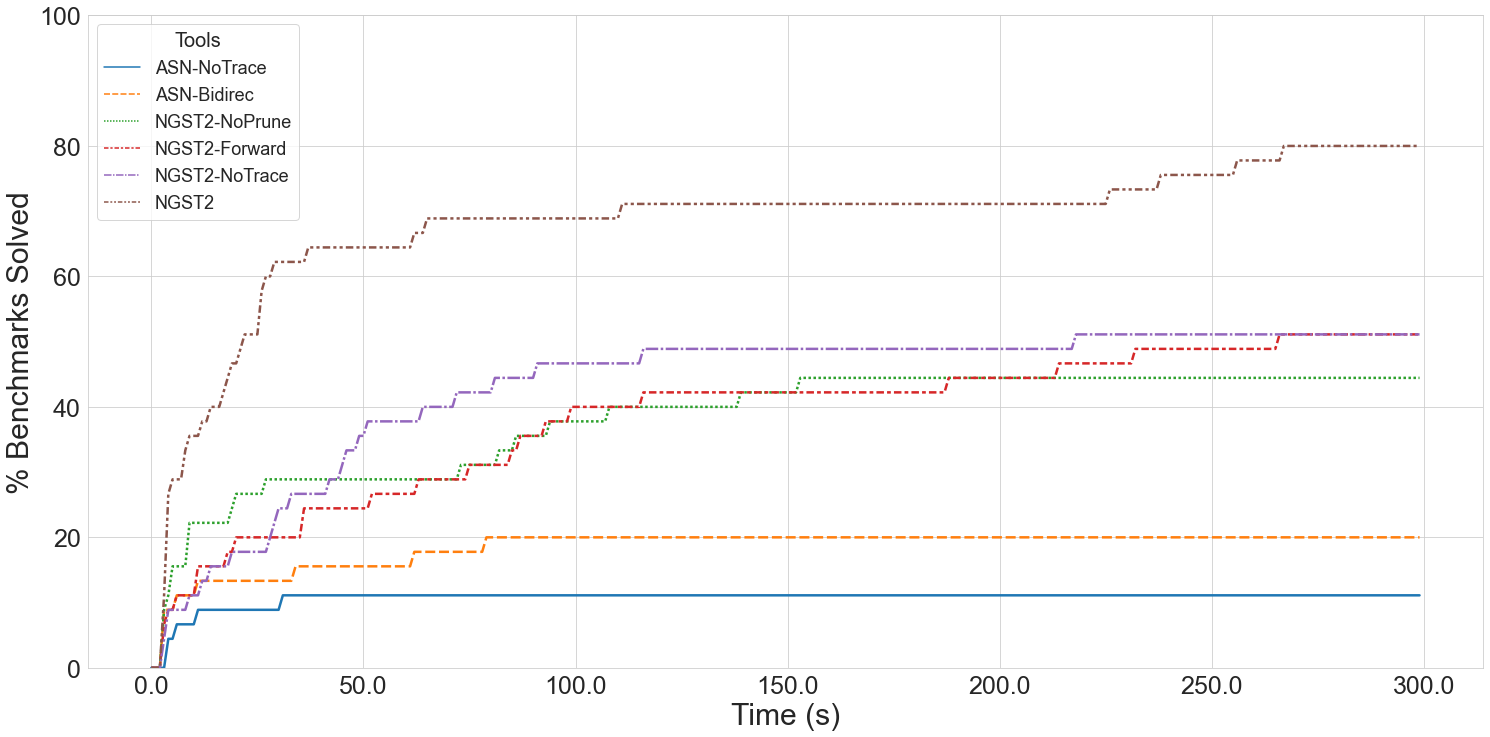}
    }
    \vspace{-0.1in}
    \caption{Performance comparison with \tool variants.}
    \label{fig:ablationGraphs}
    \ifDiff
    \vspace{-0.2in}
    \fi
\end{figure}

In this section, we describe a series of ablation studies that are designed to answer our second research question. In particular, we consider the following ablations of \tool:

\begin{itemize}[leftmargin=*]
    \item {\bf NGST2-NoPrune:} This is a variant of \tool that does not perform pruning based on trace compatibility. In particular, it does not leverage any of the techniques discussed in Section~\ref{sec:checkTraceCompat}.
    \item {\bf NGST2-Forward:} This variant of \tool uses the \emph{uni-directional} pruning rules from Section~\ref{sec:checkTraceCompat}; however, it does not utilize the backward semantics to reduce the overhead of pruning.
    \item {\bf NGST2-NoTrace:} This variant of \tool uses the \emph{bidirectional} pruning rules from Section~\ref{sec:checkTraceCompat}; however, it does not utilize \emph{intermediate} values for shared non-terminals; i.e., rule $\mathrm{S-NTerm}-\updownarrow$ simply returns a fresh variable $\nu$ like rule $\mathrm{NTerm}-\updownarrow$.
    \item {\bf ASN-Bidirec:} This ablation uses the ASN architecture instead of our proposed CGN. 
    \item {\bf ASN-NoTrace:} This variant of \tool is the same as \toolnotrace, except it uses the ASN to guide search rather than the CGN.
\end{itemize}
{Again, similar to our baseline comparison, we use the same training data, training methodology, and hyperparameter setting approach for all ablations.}


The results of these ablation studies are summarized in Figure~\ref{fig:ablationGraphs}  using the same type of cactus plot as in Figure~\ref{fig:baselineGraphs}. The main take-away from this experiment is that \tool significantly outperforms all other variants for both clients. In particular, \tool solves more benchmarks in the first $45$ seconds than any other variant solves in the total $5$ minutes. 

In terms of the relative importance of the proposed neural architecture vs. the pruning component, the data are mixed. For the Java Stream client, \toolasnbidirec outperforms  \toolnoprune, while \toolnoprune outperforms \toolasnbidirec for the Python client. It appears that, for some benchmarks, the \net is more important to performance than bidirectional pruning, while for others, pruning is more important than the \net. However, it is clear that using a combination of bidirectional pruning and the proposed CGN architecture leads to the best performance.

\vspace{4mm}
 \fbox{
\begin{minipage}{0.9\linewidth}
{\bf Result for RQ2:}
Both the pruning mechanism and the \net architecture are vital to \tool; removing either component results in  more than a 25\% reduction in benchmarks solved.
\end{minipage}
}

\subsection{Evaluation of Trace Compatibility Assumption}


To answer our third research question, we performed a manual evaluation of the benchmarks that \tool fails to synthesize within the time limit. 
We found only 3 of the 91 fail to synthesize because they do not satisfy our trace compatibility assumption.
However, as we discuss below, it is possible to solve all three benchmarks using a slight relaxation of the trace compatibility assumption. 


\begin{figure}[t]
\centering
    \begin{tabular}{cc}
    \textbf{Input} & \textbf{Output} \\
           \begin{minipage}{0.5\linewidth}
\begin{lstlisting}[linewidth=6cm, frame=none, language=Python, numbers=none] 
def lowerIn(strs, S):
    for s in strs:
        if s in S:
            return s.lower()
    return None
\end{lstlisting}
\end{minipage}
&
  \begin{minipage}{0.5\linewidth}
\begin{lstlisting}[linewidth=6.1cm,language=C,frame=none, morekeywords={map, find, filter}, numbers=none,breaklines]
find(map(filter(strs, 
                $\lambda$ s. s in S), 
         $\lambda$ s. s.lower()),
     None,
     $\lambda$ s. True)
\end{lstlisting}
\end{minipage}
    \end{tabular}
\vspace{-0.1in}
\caption{Non-trace-compatible Benchmark Example}
\label{fig:no_trace_compat}
\vspace{-0.1in}
\end{figure}

To provide intuition,  Figure~\ref{fig:no_trace_compat} shows a simplified version of one of the three benchmarks that violate trace compatibility.
 Here, the imperative function iterates through each string \texttt{s} from a list of strings \texttt{strs} and returns the lower-case version of the first string which appears in  \texttt{S}. The functional translation is shown next to the imperative version: it first uses the $\code{map}$ and $\code{filter}$ operators to get the lower case versions of all strings \texttt{s} in \texttt{S}.  Then, it leverages the $\code{find}$  combinator to return the first element of that list if one exists. The source and target programs shown in Figure~\ref{fig:no_trace_compat} are not trace complete, as the expression \texttt{True} in the target does not appear in the source program. However, if we relax the trace compatibility assumption slightly and  allow the target program to contain additional "default" constants such as True, False, and 0, then all three benchmarks can be correctly transpiled by \tool.

\vspace{4mm}
 \fbox{
\begin{minipage}{0.9\linewidth}
{\bf Result for RQ3:}
Only 3 of the 91 benchmarks violate the trace-compatibility assumption. Furthermore, if we slightly relax the trace compatibility assumption to allow default constants, then those three benchmarks can also be successfully transpiled by \tool. 
\end{minipage}
}
\section{Limitations}
\new{
Our approach has some limitations based on the assumptions we made. First, we assume the target language belongs to the meta-grammar shown in Section~\ref{sec:unidirecrules}. While in theory this could restrict the space of languages we can translate to, in practice we found this meta-grammar is expressive enough to allow all of the functional APIs we were interested in. Second, we assume that we have access to the forward and backward semantics of the target language constructs and that these semantics can be encoded in some first-order theory. We do not believe this to be a major limitation, as our method does not require the precise semantics of language constructs, and thus they can be over-approximated in some reasonable manner. For example, we encode the forward and backwards semantics for Java Stream operators via the rules shown in Figures~\ref{fig:forward-semantics} and \ref{fig:backward-semantics}. Finally, our proposed neural approach assumes there is sufficient data on which to train the network. However, we also do not see this as a considerable limitation -- we found that easily-generated synthetic training data was sufficient for all the translation problems we considered (Section~\ref{sec:implementation}).
}
\section{Related Work}
\label{sec:related}

In this section, we survey prior work that is most closely related to our proposed approach.

\paragraph{Transpilation}
A number of works use program synthesis techniques for transpilation across a variety of domains. For instance, Ahmad and Cheung \cite{ahmad2018automatically} introduce Casper, a tool for synthesizing programs in the MapReduce paradigm, using imperative source implementations \footnote{We contacted the authors about comparing our tool with Casper, however, their tool is no longer in repair.}. QBS \cite{cheung2013optimizing} uses synthesis to translate application code fragments to SQL queries and STNG \cite{kamil2016verified} translates low-level Fortran code into a high-level predicate language. Mariano et al. \cite{mariano2020demystifying} use type-directed program synthesis to generate functional summaries of loops in smart contracts. 

In addition to synthesis techniques, rule-based transpilation techniques have also been proposed. Gyori et al. \cite{gyori2013crossing} suggest rule-based translation of imperative Java to the functional Stream API while Khatchadourian et al \cite{khatchadourian2020safe} propose a similar method for automatically translating sequential Java streams into parallel streams. Similarly, MOLD \cite{radoi2014translating} uses rules to automatically translate Java programs to Apache Spark.

Another emerging approach to transpilation is based on machine translation, and such techniques have been used to translate from Java to C\# \cite{koehn2007moses}, Python2 to Python3 \cite{aggarwal2015using}, and CoffeScript to JavaScript \cite{chen2018tree}. Recently, Transcoder \cite{lachaux2020unsupervised} has shown great success in this space by using large pre-trained models for learning a variety of transilation tasks. To the best of our knowledge, we are the first to propose a hybrid  technique that  combines a neural component for guiding the search and a  pruning component that leverages intermediate values.
\new{
\paragraph{MapReduce Synthesis}
The problem solved in this work is related to the problem of MapReduce synthesis which attempts to automatically generate parallelized programs. Raychev et al. \cite{raychev2015parallelizing} propose a technique for automatically parallelizing imperative user-defined aggregations (UDAs) using symbolic execution. Farzan and Nicolet \cite{farzan2017synthesis} also address the problem of UDA parallelization, in their case by synthesizing joins which combine the results of executing portions of the UDA in parallel. In general, these approaches are limited to UDAs and are thus not applicable to the more general imperative loops considered in this project. Radoi et al. \cite{radoi2014translating} propose a technique for translating sequential loops into MapReduce programs via a reduction to the lambda calculus and then a series of rewrites. Similar to \cite{gyori2013crossing}, this approach relies on a set of hand-written rules while ours does not. Smith and Albarghouthi \cite{smith2016mapreduce} propose a technique for synthesizing MapReduce programs from input-output examples. Their approach requires a small fixed set of higher-order sketches (only 8 in their evaluation) which they use to significantly reduce the search space; defining such a set of sketches in our context is not feasible given the varied nature of the benchmarks we consider.
}

\paragraph{Pruning and Program Synthesis}
{A number of existing synthesis tools use top-down and/or bottom-up reasoning to prune the search space. Synquid \cite{polikarpova2016program} uses a combination of top-down and bottom-up reasoning to propagate refinement type constraints to unknowns during program search. Similarly, STUN \cite{alur2015synthesis} uses a combination of top-down and bottom-up reasoning to decompose a single synthesis problem into multiple subproblems whose solutions are unified together. $\lambda^2$ \cite{feser2015synthesizing} leverages a top-down deductive approach for inferring input-output examples of missing subexpressions. However, none of these approaches incorporate reasoning based on intermediate values via concolic execution as our technique does.}

\paragraph{Copy Mechanisms}
As discussed in Section~\ref{sec:net}, the \net is an extension of the \asn \cite{asn} with a copy mechanism \cite{jia-liang-2016-data, see-etal-2017-get} for copying components from the input to the output. Copy mechanisms have been effectively used in a number of settings, including document summarization \cite{see-etal-2017-get}, dialogue generation \cite{gu-etal-2016-incorporating}, and have even been used in the context of program synthesis for copying string variables \cite{jia-liang-2016-data} and predicting out-of-vocabulary tokens \cite{li-ijcai-pointer}. However, as far as we are aware, we are the first to suggest using copy mechanisms for copying large subtrees of ASTs.

\paragraph{Neurosymbolic Program Synthesis}
In recent years, neurosymbolic techniques have become quite popular. Notable applications of  neurosymbolic program synthesis include program correction \cite{bhatia2018neuro}, compositional rule-based translation \cite{nye2020learning}, rediscovering classic physics laws \cite{ellis2020dreamcoder},\old{ and} scraping website information \cite{chen2021web}\new{, and learning regular expressions \cite{ye2020optimal}}. As far as we are aware, we are the first to apply neurosymbolic synthesis techniques to the problem of imperative-to-functional API translation \new{and the only to propose symbolic reasoning based on intermediate values via concolic execution}.

\section{Conclusion and Future Work}
\label{sec:conclusion}

We have proposed a new technique, and its implementation in a tool called \tool, for translating imperative code to functional variants. Our method is based on the assumption that the overwhelming majority of  source/target programs are \emph{trace-compatible}, meaning that low-level expressions are not only syntactically shared but also take the same values in a pair of corresponding program traces. Our approach leverages this assumption to (1) design a new neural architecture called \emph{cognate grammar network (CGN}) to effectively guide transpilation and (2) develop bidirectional type checking rules that allow pruning partial programs based on intermediate values that arise during a computation. We have evaluated our approach against several baselines and ablations in the context of translating imperative Java and Python code to functional APIs targeting these languages. Our results show that the proposed techniques are both useful and necessary in this context.
More broadly, while  \tool has only been evaluated in the context of imperative-to-functional translation,  our techniques are applicable in any domain satisfying the trace-compatibility assumption, which is likely to hold in other contexts as well.  We leave it to future work to explore the applicability of our techniques (as well as their limitations) for other transpilation tasks.
\section{Acknowledgements}

We would like to thank Benjamin Sepanski, Shankara Pailoor, and Jocelyn Chen for their thoughtful feedback. This material is based upon work supported by the National Science Foundation under grant numbers CCF-1908494, CCF-1811865, CCF-1918889, DARPA under the HARDEN program, Google under the Google Faculty Research Grant, as well as both Intel and RelationalAI.

\bibliography{references}

\ifExt
\newpage

\begin{appendices}

\section{Pruning Soundness}
\label{app:pruning}

In what follows, we give the formalization and proof of the following theorem assuming the unidirectional consistency judgment (Section~\ref{subsec:uniproof}) and bidirectional consistency judgment (Section~\ref{subsec:biproof}).

\begin{theorem}
Suppose that $\nconjudge{\sigma}{\prog}{\prog'}$ and both the forward semantics $\symexec{\cdot}$ and backward semantics $\symexecback{\cdot}{-i}$ provided are sound. Then, there is no completion $\prog_t$ of $\prog'$ that is trace-compatible with $\prog$ on input $\sigma$ where $\exec{\prog_t}{\sigma} = \exec{\prog}{\sigma}$. 
\end{theorem}

We will first provide some necessary definitions and then proceed to the proofs.

\subsection{Definitions}

In addition to the definitions in the text, we also add a few other definitions, including the soundness of the forward and backward semantics.

\begin{definition}{\bf (Soundness of forwards semantics)}
\label{def:forwardSound}
Given some target language $\lang_t$ which is an instance of the meta-grammar (with first order functions $F$ and higher-order functions $F_H$) shown in \ref{sec:unidirecrules}, we say a forward semantics $\symexec{\cdot}$ is \emph{sound} wrt to $\lang_t$ if the following two conditions hold:
\begin{enumerate}
    \item For each first-order function $f \in F$, if $\exec{f}{\sigma}(e_1, \ldots, e_N) = k$, $\symexec{f}(\symexp_1, \ldots, \symexp_n) = \symexp_k$, and $\exec{e_i}{\sigma} = \symexp_i$ is satisfiable for all $1 \leq i \leq N$, then $k = \symexp_k$ is satisfiable.
    \item For each higher-order function $f \in F_H$, if $\exec{f}{\sigma}([e_1, \ldots, e_N], \{ e_1 \mapsto e_1', \ldots e_N \mapsto e_N' \}) = k$, $\symexec{f}([\symexp_1, \ldots, \symexp_N], \{ \symexp_1 \mapsto \symexp_1', \ldots \symexp_N \mapsto \symexp_N' \}) = \symexp_k$, and both $\exec{e_i}{\sigma} = \symexp_i$ and $\exec{e_i'}{\sigma[i \mapsto \exec{e_i}{\sigma}]} = \symexp_i'$ are satisfiable for all $1 \leq i \leq N$, then $k = \symexp_k$ is satisfiable.
\end{enumerate}
\end{definition}

Informally, the definition states that the forward semantics are sound if the results of symbolically executing a function $f$ on symbolic arguments always results in a symbolic value $\symexp_k$ which could equal concrete value $k$, the result of executing $f$ on any values which the symbolic arguments of $f$ could equal.



\begin{definition}{\bf (Soundness of backwards semantics)}
\label{def:backwardSound}
Given some target language $\lang_t$ which is an instance of the meta-grammar (with first order functions $F$ and higher-order functions $F_H$) shown in \ref{sec:unidirecrules}, we say a backward semantics $\symexecback{\cdot}{-i}$ is \emph{sound} wrt to a forward semantics $\symexec{\cdot}$ of $\lang_t$ if the following two conditions holds:
\begin{enumerate}
    \item For each first-order function $f \in F$, if $\symexec{f}(\symexp_1, \ldots, \symexp_N) = \symexp$ and $\varphi[\symexp]$ is satisfiable, then $\symexecback{f}{-i}(\varphi, \symexp_1, \ldots, \symexp_{i-1})[\symexp_i]$ is satisfiable for $1 \leq i \leq N$.
    \item For each higher-order fucntion $f \in F_H$, if $\symexec{f}([\symexp_1, \ldots, \symexp_k], \{ \symexp_1 \mapsto \symexp_1', \ldots, \symexp_k \mapsto \symexp_k' \}) = \symexp$ and $\varphi[\symexp]$, then $\symexecback{f}{-1}(\varphi)[[\symexp_1, \ldots, \symexp_k]]$ is satisfiable and $\symexecback{f}{-2}(\varphi, \symexp_i)[\symexp_i']$ is satisfiable for $1 \leq i \leq k$.
\end{enumerate}
\end{definition}

Informally, the definition states that the backward semantics are sound if the constraints resulting from backwards execution of $f$ on symbolic arguments and formula $\varphi$ are always satisfiable given that forwards executing $f$ on those same symbolic arguments produces a symbolic value which satisfies $\varphi$.

\begin{definition}{\bf (Concretized Valuations)}
\label{def:concMap}
Given some valuation $\sigma$ potentially containing symbolic values, we define the concretizations of $\sigma$, denoted $\mathsf{Conc}(\sigma)$, to be the set of all mappings $\sigma'$ such that for every mapping $[i \mapsto \symexp_i]$ in $\sigma$, $\sigma'$ has some mapping $[i \mapsto v_i]$ where $v_i = \symexp_i$ is satisfiable.
\end{definition}


\begin{definition}{\bf (Subexpression Trace Compatible)}
We say partial expression $\pexp$ is \emph{subexpression trace compatible} with $\prog$ on $\sigma$ if there exists some partial program $\prog'$ such that:
\begin{enumerate}
    \item $\pexp \in \prog'$
    \item $\prog' \derives \prog_t$
    \item $\prog_t$ is trace compatible with $\prog$ on $\sigma$
\end{enumerate}
\end{definition}

We will additionally define a collecting semantics for \emph{contexts} (as opposed to values) $\collctx{\cdot}{\prog}{\sigma}$ which maps an expression $e$ to the set of contexts $\sigma'$ in which $e$ is executed during an execution of $\prog$ on $\sigma$.

\subsection{Unidirectional Pruning Soundness Proof}
\label{subsec:uniproof}

To prove unidirectional pruning soundness, we first give the full set of unidirectional pruning rules, then prove a useful lemma, and finally prove the full theorem.

\subsubsection{Unidirectional Pruning Rules}

\begin{figure}
    \centering
    \textbf{Uni-directional Trace-Compatibility Checking Rules $\uparrow$} \\
    \begin{mathpar}
    \inferrule*[Right=$\nconsist$-$\uparrow$, width=10cm]{
       \sigma \not \vdash \prog' \consist \prog
    } {
       \nconjudge{\sigma}{\prog'}{\prog}
    }
    
    \inferrule*[Right=$\consist$-$\uparrow$, width=10cm]{
      \upjudge{\prog, \sigma}{\prog'}{\symexp} \\ \mathsf{SAT}(\symexp = \exec{\prog}{\sigma}) 
    } {
      \conjudge{\sigma}{\prog'}{\prog}
    }
    \\
    \inferrule*[Right=NTerm-$\uparrow$, width=10cm]{
       \neg \mathrm{Shared}(N) \\ \mathrm{FreshVar}(\nu)
    } {
       \upjudge{\prog, \sigma}{N}{\nu}
    }
    \\
    \inferrule*[Right=S-NTerm-$\uparrow$, width=10cm]{
       \mathrm{Shared}(N)  \\
       \mathrm{FreshVar}(\nu) \\
       \mathsf{SAT}(\bigvee_{w \in \prog, \ N \derives w}\bigvee_{c \in \coll{w}{\prog}{\sigma}} \nu = c)
    } {
       \upjudge{\prog, \sigma}{N}{\nu}
    }
    \\
    \inferrule*[Right=VarTerm-$\uparrow$, width=10cm]{
       \mathrm{Variable}(v) \\ 
       \sigma[v] = \symexp
    } {
       \upjudge{\prog, \sigma}{v}{\symexp}
    }
    \\
    \inferrule*[Right=FirstOrderTerm-$\uparrow$, width=10cm]{
       \mathrm{Function}(f) \\ 
       \mathrm{FirstOrder}(f) \\
       \forall i. \, \upjudge{\prog, \sigma}{\pexp_i}{\symexp_i} \\ 
       \symexec{f}(\symexp_1, ..., \symexp_n) = \symexp
    } {
       \upjudge{\prog, \sigma}{f(\pexp_1, ..., \pexp_n)}{\symexp}
    }
    \\
    \inferrule*[Right=HigherOrderTerm-$\uparrow$, width=10cm]{
       \mathrm{Function}(f) \\
       \mathrm{HigherOrder}(f) \\
       \upjudge{\prog, \sigma}{\pexp_1}{[\symexp_1, ..., \symexp_k]} \\ 
       \upjudge{\prog, \sigma[i \mapsto \symexp_i]}{\pexp_2}{\symexp_i'} \\
       \symexec{f}([\symexp_1, ..., \symexp_k], \{ \symexp_1 \mapsto \symexp_1', ..., \symexp_k \mapsto \symexp_k' \}) = \symexp
    } {
       \upjudge{\prog, \sigma}{f(\pexp_1, \, \lambda i. \, \pexp_2)}{\symexp}
    }
    
    \end{mathpar}
    \caption{Uni-directional rules for checking trace compatibility between partial program in target language and a program in the source language. We assume that lambda bindings in higher-order functions do not shadow input variables. }
    \label{fig:forwardrulesext}
    \vspace{-0.1in}
\end{figure}

The full set of unidirectional pruning rules is shown in Figure~\ref{fig:forwardrulesext}. Note, we slightly changed the rule \textsc{S-NTerm-}$\uparrow$ from the presentation in the main text. In particular, we made the rule deterministic (and thus the algorithm deterministic) by having the rule return a symbolic value capturing \emph{all} possible values for a completion of the nonterminal; in contrast, the original rule nondeterministically chose one such possible value. The resulting deterministic version has equivalent pruning power to the rule shown in the main text, but this version is easier to use to prove soundness, so we assume this set of rules going forward in the proof.

\subsubsection{Unidirectional Inclusion Lemma}

We now state and prove the \emph{unidirectional inclusion lemma} which will enable us to prove unidirectional pruning soundness.

\begin{lemma}{(Unidirectional Inclusion Lemma)}
\label{lemma:uniinclusion}
Given some partial expression $\pexp$ and symbolic value $\symexp$, if $\upjudge{\prog, \sigma'}{\pexp}{\symexp}$ where $\pexp$ is subexpression trace compatible with $\prog$ on $\sigma$, then it is the case that for every completion $e$ of $\pexp$ and valuation $\sigma''$ that is both a concretization of $\sigma'$ and a context of $e$ in $\prog_t$ on $\sigma$, the formula $\symexp = \exec{e}{\sigma''}$ is satisfiable.
\end{lemma}

Informally, this lemma states that if the unidirectional rules generate some symbolic value $\symexp$ for partial expression $\pexp$, than that symbolic value is an \emph{overapproximation} of the possible values a completion $e$ of $\pexp$ could take on during execution of a trace-compatible program $\prog_t$.

\begin{proof}
We proceed with an inductive proof, starting with the three base cases: \\

{\bf Base Case 1: $\pexp = N$ (unshared)}. In this case, the lemma trivially holds, as for \textsc{NTerm-}$\uparrow$, $\symexp$ is simply a fresh variable, and thus $\mathsf{SAT}(\exec{e}{\sigma''} = \symexp)$ for all $\sigma''$.

{\bf Base Case 2: $\pexp = v$}. In this case, the lemma trivially holds, as for \textsc{VarTerm-}$\uparrow$, $\symexp$ is simply a $\sigma[v]$, and thus by the definition of a concretization, for any $\sigma'' \in \mathsf{Conc}(\sigma)$ it is the case that $\mathsf{SAT}(\exec{e}{\sigma''} = \symexp)$.

{\bf Base Case 3: $\pexp = N$ (shared)}. By rule \textsc{S-NTerm-}$\uparrow$, we know it is the case that $\symexp$ is a fresh variable $\nu$ constrained by the following: 
\begin{equation}
\label{eq:sntermrule}
\mathsf{SAT}(\bigvee_{w \in \prog, \ N \derives w}\bigvee_{c \in \coll{w}{\prog}{\sigma}} \nu = c)    
\end{equation}
For any completion $e$ of $\pexp$, let us define 
\begin{equation}
\label{eq:sntermsedef}
S_e = \{ \exec{e}{\sigma''} \ | \ \sigma'' \in \collctx{e}{\prog_t}{\sigma} \land \sigma'' \in \mathsf{Conc}(\sigma') \}    
\end{equation}
By the definition of the collecting semantics and the trace compatibility assumption, we know 
\begin{equation}
\label{eq:sntermtracecompat}
S_e \subseteq \coll{e}{\prog_t}{\sigma} \subseteq \coll{e}{\prog}{\sigma}    
\end{equation}
Using this and the constraint on $\symexp$, we see that for all $v \in S_e$, $\mathsf{SAT}(v = \symexp)$, and thus the lemma holds. \\

{\bf Inductive Hypothesis:} The lemma holds for all (potentially partial) subexpressions. \\

{\bf Case 1: $\pexp = f(\pexp_1, \ldots, \pexp_n)$}. In this case, we first apply the inductive hypothesis and the judgement $\upjudge{\prog, \sigma'}{\pexp_i}{\symexp_i}$ from \textsc{FirstOrderTerm-}$\uparrow$ to find the following valid formula:
\begin{equation}
\label{eq:fotindhyp}
    \forall \sigma_i \in \collctx{e_i}{\prog_t}{\sigma}. \ (\sigma_i \in \mathsf{Conc}(\sigma')) \rightarrow \mathsf{SAT}(\exec{e_i}{\sigma_i} = \symexp_i)
\end{equation}
Then, using \ref{eq:fotindhyp} and the soundness of $\symexec{\cdot}$ (along with the fact that $\symexec{f}(\symexp_1, \ldots, \symexp_n) = \symexp$), we can conclude the lemma holds, that is, that the following is valid:
\begin{equation}
    \forall \sigma'' \in \collctx{e}{\prog_t}{\sigma}. \ (\sigma'' \in \mathsf{Conc}(\sigma')) \rightarrow \mathsf{SAT}(\exec{e}{\sigma''} = \symexp)
\end{equation}
Critically, applying the $\symexec{\cdot}$ soundness definition requires that we know the following:
\begin{equation}
    \sigma'' \in \collctx{e}{\prog_t}{\sigma} \Rightarrow \sigma'' \in \collctx{e_i}{\prog_t}{\sigma}
\end{equation}
This is satisfied assuming a call by value semantics, which we assume in this context.

{\bf Case 2: $\pexp = f(\pexp_1, \lambda i. \pexp_2)$}. In this case, we first apply the inductive hypothesis and the judgement $\upjudge{\prog, \sigma'}{\pexp_1}{[\symexp_1, \ldots, \symexp_n]}$ from \textsc{HigherOrderTerm-}$\uparrow$ to find the following valid formula:
\begin{equation}
\label{eq:hotindhype1}
    \forall \sigma_{e1} \in \collctx{e_1}{\prog_t}{\sigma}. \ (\sigma_{e1} \in \mathsf{Conc}(\sigma')) \rightarrow \mathsf{SAT}(\exec{e_1}{\sigma_{e1}} = [\symexp_1, \ldots, \symexp_n])
\end{equation}
We then apply the inductive hypothesis to the judgement $\upjudge{\prog, \sigma'[i \mapsto \symexp_i]}{\pexp_2}{\symexp_i}$ to get the following valid formula:
\begin{equation}
\label{eq:hotindhype2}
    \forall \sigma_i \in \collctx{e_2}{\prog_t}{\sigma}. \ (\sigma_i \in \mathsf{Conc}(\sigma'[i \mapsto \symexp_i])) \rightarrow \mathsf{SAT}(\exec{e_2}{\sigma_i} = \symexp_i)
\end{equation}
Then, using \ref{eq:hotindhype1} and \ref{eq:hotindhype2} and the soundness of $\symexec{\cdot}$ (along with the fact that $\symexec{f}([\symexp_1, \ldots, \symexp_k], \{ \symexp_1 \mapsto \symexp_1', \ldots, \symexp_k \mapsto \symexp_k' \}) = \symexp$), we can conclude that the lemma holds, that is, that the following is valid:
\begin{equation}
    \forall \sigma'' \in \collctx{e}{\prog_t}{\sigma}. \ (\sigma'' \in \mathsf{Conc}(\sigma')) \rightarrow \mathsf{SAT}(\exec{e}{\sigma''} = \symexp)
\end{equation}
Critically, applying the $\symexec{\cdot}$ soundness definition requires that we know the following:
\begin{equation}
    \sigma'' \in \collctx{e}{\prog_t}{\sigma} \Rightarrow (\sigma'' \in \collctx{e_1}{\prog_t}{\sigma} \land \sigma''[i \mapsto \exec{e_1}{\sigma''}[i]] \in \collctx{e_2}{\prog_t}{\sigma})
\end{equation}
This is satisfied assuming a call by value semantics, which we assume in this context.

\end{proof}

\subsubsection{Soundness Proof}

We give a short proof of the main theorem here, which is easily obtained from the inclusion lemma.

\begin{proof}
We proceed with the proof of soundness by way of contradiction.

In particular, BWOC, suppose there exist some $\prog$, $\prog'$, $\sigma$, and $\prog_t$ such that:

\begin{enumerate}
    \item[A1.] $\nconjudge{\sigma}{\prog'}{\prog}$
    \item[A2.] $\prog' \derives \prog_t$
    \item[A3.] $\prog_t$ is trace-compatible with $\prog$ on valuation $\sigma$.
    \item[A4.] $\exec{\prog_t}{\sigma} = \exec{\prog}{\sigma}$
    \item[A5.] The forward and backward semantics are sound.
\end{enumerate}

Given these assumptions, we begin as follows:

By A1 and rule $\nconsist$\textsc{-}$\uparrow$, we conclude 
\begin{equation}
\sigma \not \vdash \prog' \consist \prog    
\end{equation}
Using this and rule $\consist$\textsc{-}$\uparrow$, we conclude there is no symbolic expression $\symexp$ such that 
\begin{equation}
\label{eq:unisateq}
\upjudge{\prog, \sigma}{\prog'}{\symexp} \land \mathsf{SAT}(\symexp = \exec{\prog}{\sigma})
\end{equation}
Given the definition of the $\uparrow$ judgment, we know there is a symbolic expression $\symexp$ which satisfies this judgement. Thus, using the Inclusion Lemma (Lemma~\ref{lemma:uniinclusion}) and the fact that $\sigma \in \collctx{\prog_t}{\prog_t}{\sigma}$, we know
\begin{equation}
    \mathsf{SAT}(\exec{\prog_t}{\sigma} = \symexp)
\end{equation}
Thus, by A4, we find a contradiction with \ref{eq:unisateq}.

\end{proof}

\subsection{Bidirectional Pruning Soundness Proof}
\label{subsec:biproof}

The proof of bidirectional pruning soundness builds off of the unidirectional pruning soundness proof. In the following section, we first give the full set of bidirectional pruning rules, define and prove some useful lemmas, and then give the proof of bidirectional pruning soundness.

\subsubsection{Bidirectional Pruning Rules}

\begin{figure}
    \centering
    \textbf{Bidirectional Trace-Compatibility Checking Rules $\updownarrow$} \\
    \begin{mathpar}
    \inferrule*[Right=$\nconsist$-$\updownarrow$, width=10cm]{
      \sigma \not \vdash \prog' \consist \prog
    } {
      \nconjudge{\sigma}{\prog'}{\prog}
    }
    
    \inferrule*[Right=$\consist$-$\updownarrow$, width=10cm]{
      \udjudge{\prog, \sigma\textcolor{blue}{, y = \exec{\prog}{\sigma}}}{\prog'}{\symexp} \\
      \mathsf{SAT}(\symexp = \exec{\prog}{\sigma})
    } {
      \conjudge{\sigma}{\prog'}{\prog}
    }
    \\
    \inferrule*[Right=NTerm-$\updownarrow$, width=10cm]{
      \neg \mathrm{Shared}(N) \\ 
      \mathrm{FreshVar}(\nu) \\ 
      \textcolor{blue}{\mathsf{SAT}(\varphi[\nu])}
    } {
      \udjudge{\prog, \sigma\textcolor{blue}{, \varphi}}{N}{\nu}
    }
    \\
    \inferrule*[Right=S-NTerm-$\updownarrow$, width=10cm]{
      \mathrm{Shared}(N)  \\
      \mathrm{FreshVar}(\nu) \\
      \mathsf{SAT}(\bigvee_{w \in \prog, \ N \derives w}\bigvee_{c \in \coll{w}{\prog}{\sigma}} \nu = c) \\
      \textcolor{blue}{\mathsf{SAT}(\varphi[\nu])}
    } {
      \udjudge{\prog, \sigma\textcolor{blue}{, \varphi}}{N}{\nu}
    }
    \\
    \inferrule*[Right=VarTerm-$\updownarrow$, width=10cm]{
      \mathrm{Variable}(v) \\ 
      \sigma[v] = \symexp \\ 
      \textcolor{blue}{\mathsf{SAT}(\varphi[\symexp])}    
    } {
      \udjudge{\prog, \sigma\textcolor{blue}{, \varphi}}{v}{\symexp}
    }
    \\
    \inferrule*[Right=FirstOrderTerm-$\updownarrow$, width=10cm]{
      \mathrm{Function}(f) \\ 
      \mathrm{FirstOrder}(f) \\
      \forall i. \, \udjudge{\prog, \sigma\textcolor{blue}{, \symexecback{f}{-i}(\varphi, \symexp_1, ..., \symexp_{i-1})}}{\pexp_i}{\symexp_i} \\
      \symexec{f}(\symexp_1, ..., \symexp_n) = \symexp \\ 
      \textcolor{blue}{\mathsf{SAT}(\varphi[\symexp])}
    } {
      \udjudge{\prog, \sigma\textcolor{blue}{, \varphi}}{f(\pexp_1, ..., \pexp_n)}{\symexp}
    }
    \\
    \inferrule*[Right=HigherOrderTerm-$\updownarrow$, width=8cm]{
      \mathrm{Function}(f) \\ 
      \mathrm{HigherOrder}(f) \\
      \udjudge{\prog, \sigma\textcolor{blue}{, \symexecback{f}{-1}(\varphi)}}{\pexp_1}{[\symexp_1, ..., \symexp_k]} \\
      \udjudge{\prog, \sigma[i \mapsto \symexp_i]\textcolor{blue}{, \symexecback{f}{-2}(\varphi, \symexp_i)}}{\pexp_2}{\symexp_i'} \\
      \symexec{f}([\symexp_1, ..., \symexp_k], \{ \symexp_1 \mapsto \symexp_1', ..., \symexp_k \mapsto \symexp_k' \}) = \symexp \\
      \textcolor{blue}{\mathsf{SAT}(\varphi[\symexp])}
    } {
      \udjudge{\prog, \sigma\textcolor{blue}{, \varphi}}{f(\pexp_1, \, \lambda i. \, \pexp_2)}{\symexp}    
    }    
    
    \end{mathpar}
    \caption{Bidrectional Rules}
    \label{fig:bidirecrulesext}
    \vspace{-0.2in}
\end{figure}

The full set of bidirectional pruning rules is shown in Figure~\ref{fig:bidirecrulesext}. In the same manner as the unidirectional rules, we slightly changed the rule \textsc{S-NTerm-}$\uparrow$ to make the rule deterministic and thus easier to reason about. Again, this results in equivalent pruning power to the rule shown in the main text.

\subsubsection{Pruning Equivalence Lemmas}

To prove the soundness of the bidirectional pruning rules, we will simply prove that the pruning power of the unidirectional pruning rules is equivalent to the pruning power of the bidirectional pruning rules. Recall from the main text that the two have the same (theoretical) pruning power, but that the unidirectional rules scale more poorly. To establish the equivalence of the pruning power, we will rely on the following two lemmas relating the $\uparrow$-judgments and the $\updownarrow$-judgments.

\begin{lemma}
\label{lemma:biImpUni}
For any programs $\prog$ and $\prog'$, valuation $\sigma$, formula $\varphi$, and symbolic value $\symexp$ for which we have $\udjudge{\prog, \sigma, \varphi}{\prog'}{\symexp}$, we also have $\upjudge{\prog, \sigma}{\prog'}{\symexp}$.
\end{lemma}

Informally, this rule states that any symbolic value satisfying the $\updownarrow$-judgment will also satisfy the matching $\uparrow$-judgment.

\begin{proof}
We proceed with an inductive proof, starting with the three base cases: \\

{\bf Base Case 1: $\pexp = N$ (unshared)}. Trivial, as both \textsc{NTerm-}$\uparrow$ and \textsc{NTerm-}$\updownarrow$ produce fresh variable $\nu$ with no additional restrictions from \textsc{Nterm-}$\uparrow$.

{\bf Base Case 2: $\pexp = v$}. Trivial, as both \textsc{VarTerm-}$\uparrow$ and \textsc{VarTerm-}$\updownarrow$ produce symbolic expression $\sigma[v]$ with no additional restrictions from \textsc{Varterm-}$\uparrow$.

{\bf Base Case 3: $\pexp = N$ (shared)}. Trivial, as both \textsc{S-NTerm-}$\uparrow$ and \textsc{S-NTerm-}$\updownarrow$ produce fresh variable $\nu$ constrained by the formula $\mathsf{SAT}(\bigvee_{w \in \prog, \ N \derives w}\bigvee_{c \in \coll{w}{\prog}{\sigma}} \nu = c)$ with no additional restrictions from \textsc{S-NTerm-}$\uparrow$. \\

{\bf Inductive Hypothesis:} The lemma holds for all (potentially partial) subexpressions. \\

{\bf Case 1: $\pexp = f(\pexp_1, \ldots, \pexp_n)$}. Using the lemma, we assume
\begin{equation}
    \udjudge{\prog, \sigma, \varphi}{f(\pexp_1, \ldots, \pexp_n)}{\symexp}
\end{equation}
Using the rule \textsc{FirstOrderTerm-}$\updownarrow$, we use this to conclude the following for $1 \leq i \leq n$
\begin{equation}
    \udjudge{\prog, \sigma, \symexecback{f}{-i}(\varphi, \symexp_1, \ldots, \symexp_{i-1})}{\pexp_i}{\symexp_i}
\end{equation}
Using the inductive hypothesis, we can conclude the following for $1 \leq i \leq n$
\begin{equation}
    \upjudge{\prog, \sigma}{\pexp_i}{\symexp_i}
\end{equation}
Finally, combining this with the fact that $\symexec{f}(\symexp_1, \ldots, \symexp_n) = \symexp$ (from the assumption of the lemma that the $\updownarrow$-judgment holds), we find that the RHS of the lemma holds, that is
\begin{equation}
    \upjudge{\prog, \sigma}{f(\pexp_1, \ldots, \pexp_n)}{\symexp}
\end{equation}
\\

{\bf Case 2: $\pexp = f(\pexp_1, \lambda i. \pexp_2)$}. Using the lemma, we assume
\begin{equation}
    \label{eq:biImpUniCase2Assump}
    \udjudge{\prog, \sigma, \varphi}{f(\pexp_1, \lambda i. \pexp_2)}{\symexp}
\end{equation}
Using the rule \textsc{HigherOrderTerm-}$\updownarrow$, we use this to conclude the following
\begin{equation}
    \label{eq:biImpUniCase2r1}
    \udjudge{\prog, \sigma, \symexecback{f}{-1}(\varphi)}{\pexp_1}{[\symexp_1, \ldots, \symexp_k]}
\end{equation}
Furthermore, using \ref{eq:biImpUniCase2Assump}, we can conclude for $1 \leq i \leq k$
\begin{equation}
    \label{eq:biImpUniCase2r2}
    \udjudge{\prog, \sigma, \symexecback{f}{-2}(\varphi, \symexp_i)}{\pexp_2}{\symexp_i'}
\end{equation}
We use the inductive hypothesis with \ref{eq:biImpUniCase2r1} to get
\begin{equation}
    \label{eq:biImpUniCase2r1ind}
    \upjudge{\prog, \sigma}{\pexp_1}{[\symexp_1, \ldots, \symexp_k]}
\end{equation}
and the inductive hypothesis with \ref{eq:biImpUniCase2r2} to get
\begin{equation}
    \label{eq:biImpUniCase2r2ind}
    \upjudge{\prog, \sigma}{\pexp_2}{\symexp_i'}
\end{equation}
Finally, combining \ref{eq:biImpUniCase2r1ind} and \ref{eq:biImpUniCase2r2ind} with the fact that $\symexec{f}([\symexp_1, \ldots, \symexp_n], \{ \symexp_1 \mapsto \symexp_1', \ldots, \symexp_k \mapsto \symexp_k' \}) = \symexp$ (from the assumption of the lemma that the $\updownarrow$-judgment holds), we find that the RHS of the lemma holds, that is
\begin{equation}
    \upjudge{\prog, \sigma}{f(\pexp_1, \lambda i. \pexp_2)}{\symexp}
\end{equation}

\end{proof}

\begin{lemma}
\label{lemma:uniImpBi}
For any programs $\prog$ and $\prog'$, valuation $\sigma$, and symbolic value $\symexp$ for which we have $\upjudge{\prog, \sigma}{\prog'}{\symexp}$, we also have $\udjudge{\prog, \sigma, \varphi}{\prog'}{\symexp}$ for all formulas $\varphi$ such that $\mathsf{SAT}(\varphi[\symexp])$.
\end{lemma}

Informally, this rule states that any symbolic value $\symexp$ satisfying the $\uparrow$-judgment will also satisfy the matching $\updownarrow$-judgment as long as the formula $\varphi$ added to the $\updownarrow$-judgment is satisfiable on $\symexp$.

\begin{proof}
We proceed with an inductive proof, starting with the three base cases: \\

{\bf Base Case 1: $\pexp = N$ (unshared)}. Trivial, as both \textsc{NTerm-}$\updownarrow$ and \textsc{NTerm-}$\uparrow$ produce fresh variable $\nu$ with the only additional restriction from \textsc{Nterm-}$\updownarrow$ being that $\mathsf{SAT}(\varphi[\symexp])$ which we already assume of $\varphi$ in the lemma.

{\bf Base Case 2: $\pexp = v$}. Trivial, as both \textsc{VarTerm-}$\updownarrow$ and \textsc{VarTerm-}$\uparrow$ produce the symbolic value $\sigma[v]$ with the only additional restriction from \textsc{Varterm-}$\updownarrow$ being that $\mathsf{SAT}(\varphi[\symexp])$ which we already assume of $\varphi$ in the lemma.

{\bf Base Case 3: $\pexp = N$ (shared)}. Trivial, as both \textsc{S-NTerm-}$\updownarrow$ and \textsc{S-NTerm-}$\uparrow$ produce fresh variable $\nu$ constrained by the formula $\mathsf{SAT}(\bigvee_{w \in \prog, \ N \derives w}\bigvee_{c \in \coll{w}{\prog}{\sigma}} \nu = c)$ with the only additional restriction from \textsc{S-NTerm-}$\updownarrow$ being that $\mathsf{SAT}(\varphi[\nu])$ which we already assume of $\varphi$ in the lemma. \\

{\bf Inductive Hypothesis:} The lemma holds for all (potentially partial) subexpressions. \\

{\bf Case 1: $\pexp = f(\pexp_1, \ldots, \pexp_n)$}. Using the lemma, we assume the following
\begin{equation}
    \upjudge{\prog, \sigma}{f(\pexp_1, \ldots, \pexp_n)}{\symexp}
\end{equation}
Using the rule \textsc{FirstOrderTerm}-$\uparrow$, we can conclude the following for $1 \leq i \leq n$
\begin{equation}
    \label{eq:uniImpBiCase1Assume1}
    \upjudge{\prog, \sigma}{\pexp_i}{\symexp_i}
\end{equation}
Similarly, we can conclude
\begin{equation}
    \label{eq:uniImpBiCase1Assume2}
    \symexec{f}(\symexp_1, \ldots, \symexp_n) = \symexp    
\end{equation}
Furthermore, let us only consider $\varphi$ which satisfy the following
\begin{equation}
    \label{eq:uniImpBiCase1Assume3}
    \mathsf{SAT}(\varphi[\symexp])
\end{equation}
We make this assumption because otherwise the lemma trivially holds. Given \ref{eq:uniImpBiCase1Assume2} and \ref{eq:uniImpBiCase1Assume3}, along with the soundness of the backwards semantics, we can conclude the following for $1 \leq i \leq n$
\begin{equation}
    \label{eq:uniImpBiCase1Sound}
    \mathsf{SAT}(\symexecback{f}{-i}(\varphi, \symexp_1, \ldots, \symexp_{i-1})[\symexp_i])
\end{equation}
Using \ref{eq:uniImpBiCase1Sound}, \ref{eq:uniImpBiCase1Assume1}, and the inductive hypothesis, we prove that
\begin{equation}
    \udjudge{\prog, \sigma, \symexecback{f}{-i}(\varphi, \symexp_1, \ldots, \symexp_{i-1})}{\pexp_i}{\symexp_i}
\end{equation}
Using this in combination with \ref{eq:uniImpBiCase1Assume2} and \ref{eq:uniImpBiCase1Assume3}, we show that the lemma holds, that is
\begin{equation}
    \udjudge{\prog, \sigma, \varphi}{f(\pexp_1, \ldots, \pexp_n)}{\symexp}
\end{equation}
for any satisfiable $\varphi$. \\

{\bf Case 2: $\pexp = f(\pexp_1, \lambda i. \pexp_2)$}. Using the lemma, we assume the following
\begin{equation}
    \upjudge{\prog, \sigma}{f(\pexp_1, \lambda i. \pexp_2)}{\symexp}
\end{equation}
Using the rule \textsc{HigherOrderTerm-}$\uparrow$, we can conclude the following
\begin{equation}
    \label{eq:uniImpBiCase2Assume1}
    \upjudge{\prog, \sigma}{\pexp_1}{[\symexp_1, \ldots, \symexp_k]}
\end{equation}
Furthermore, we can conclude the following for $1 \leq i \leq k$
\begin{equation}
    \label{eq:uniImpBiCase2Assume2}
    \upjudge{\prog, \sigma[i \mapsto \symexp_i]}{\pexp_2}{\symexp_i'}
\end{equation}
Additionally, still from \textsc{HigherOrderTerm-}$\uparrow$, we conclude
\begin{equation}
    \label{eq:uniImpBiCase2Assume3}
    \symexec{f}([\symexp_1, \ldots, \symexp_k], \{ \symexp_1 \mapsto \symexp_1', \ldots, \symexp_k \mapsto \symexp_k' \})
\end{equation}
Furthermore, let us only consider $\varphi$ which satisfy the following
\begin{equation}
    \label{eq:uniImpBiCase2Assume4}
    \mathsf{SAT}(\varphi[\symexp])
\end{equation}
We make this assumption because otherwise the lemma trivially holds. Given \ref{eq:uniImpBiCase2Assume3} and \ref{eq:uniImpBiCase2Assume4} and the soundness of the backwards semantics, we can conclude the two following facts
\begin{equation}
    \label{eq:uniImpBiCase2Sound1}
    \mathsf{SAT}(\symexecback{f}{-1}(\varphi)[[\symexp_1, \ldots, \symexp_k]])
\end{equation}
and
\begin{equation}
    \label{eq:uniImpBiCase2Sound2}
    \mathsf{SAT}(\symexecback{f}{-2}(\varphi, \symexp_i)[\symexp_i'])
\end{equation}
for $1 \leq i \leq k$. Using \ref{eq:uniImpBiCase2Sound1}, \ref{eq:uniImpBiCase2Assume1}, and the inductive hypothesis, we can conclude
\begin{equation}
    \label{eq:uniImpBiCase2Ind1}
    \udjudge{\prog, \sigma, \symexecback{f}{-1}(\varphi)}{\pexp_1}{[\symexp_1, \ldots, \symexp_k]}
\end{equation}
Similarly, using \ref{eq:uniImpBiCase2Sound2}, \ref{eq:uniImpBiCase2Assume2}, and the inductive hypothesis, we can conclude
\begin{equation}
    \label{eq:uniImpBiCase2Ind2}
    \udjudge{\prog, \sigma[i \mapsto \symexp_i], \symexecback{f}{-2}(\varphi, \symexp_i)}{\pexp_2}{\symexp_i'}
\end{equation}
for $1 \leq i \leq k$. Finally, we use \ref{eq:uniImpBiCase2Ind1}, \ref{eq:uniImpBiCase2Ind2}, \ref{eq:uniImpBiCase2Assume3}, and \ref{eq:uniImpBiCase2Assume4} to prove the lemma, that is, that the following
\begin{equation}
    \udjudge{\prog, \sigma, \varphi}{f(\pexp_1, \lambda i. \pexp_2)}{\symexp}    
\end{equation}
for any satisfiable $\varphi$.

\end{proof}

\subsubsection{Soundness Proof}

Now, we will proceed to show soundness by proving the pruning power of the unidirectional and bidirectional pruning rules is equivalent.  

\begin{proof}

Formally, we want to show that
\begin{equation}
    \sigma \vdash_\uparrow \prog' \nconsist \prog \Leftrightarrow \sigma \vdash_\updownarrow \prog' \nconsist \prog
\end{equation}
Note, we use $\vdash_\uparrow$ and $\vdash_\updownarrow$ to denote using the unidirectional and bidirectional rules respectively. We can restate this as follows
\begin{equation}
    \sigma \vdash_\uparrow \prog' \consist \prog \Leftrightarrow \sigma \vdash_\updownarrow \prog' \consist \prog
\end{equation}
Both directions fall straight out of the lemmas proved above. In particular, the first direction
\begin{equation}
    \sigma \vdash_\uparrow \prog' \consist \prog \Rightarrow \sigma \vdash_\updownarrow \prog' \consist \prog
\end{equation}
is proved by Lemma~\ref{lemma:uniImpBi}. In particular, by the LHS, we know $\upjudge{\prog,\sigma}{\prog'}{\symexp}$ and $\mathsf{SAT}(\symexp = \exec{\prog}{\sigma})$ (which is equivalent to $\mathsf{SAT}(\varphi[\symexp])$ for $\varphi \equiv y = \exec{\prog}{\sigma}$). Thus, using Lemma~\ref{lemma:uniImpBi}, we show that $\udjudge{\prog, \sigma, y = \exec{\prog}{\sigma}}{\prog'}{\symexp}$ and $\mathsf{SAT}(\symexp = \exec{\prog}{\sigma})$, showing that the RHS holds.
For the second direction, namely
\begin{equation}
    \sigma \vdash_\updownarrow \prog' \consist \prog \Rightarrow \sigma \vdash_\uparrow \prog' \consist \prog
\end{equation}
the reasoning proceeds the same, just with Lemma~\ref{lemma:biImpUni}. In particular, by the LHS, we know $\udjudge{\prog, \sigma, y  = \exec{\prog}{\sigma}}{\prog'}{\symexp}$ and $\mathsf{SAT}(\symexp = \exec{\prog}{\sigma})$. Using Lemma~\ref{lemma:biImpUni}, we show that $\upjudge{\prog, \sigma}{\prog'}{\symexp}$. Thus, combined with $\mathsf{SAT}(\symexp = \exec{\prog}{\sigma})$, we know the RHS is satisfied.

\end{proof}

\end{appendices}

\fi

\end{document}